\documentclass{article}

\usepackage{amsmath,amsthm,color}
\usepackage{amssymb}  % assumes amsmath package installed
\usepackage{amsfonts}
\usepackage{graphicx}
\usepackage{dsfont}
\usepackage{psfrag}
\usepackage{subfigure,epstopdf}
\usepackage{multirow}
\usepackage{cite}
\usepackage{bbm}
\usepackage{savesym}
\usepackage{float}
\usepackage{etoolbox}
\usepackage[font=small]{caption}

\newcommand{\numberthis}{\refstepcounter{equation}\tag{\arabic{equation}}}
\expandafter\appto\csname align*\endcsname{\numberthis}
\allowdisplaybreaks

\newtheorem{corollary}{Corollary}

\newtheorem{proposition}{Proposition}
\newtheorem{lemma}{Lemma}

\DeclareMathOperator*{\fmin}{{min^f}}

\begin{document}
%
% paper title
% Titles are generally capitalized except for words such as a, an, and, as,
% at, but, by, for, in, nor, of, on, or, the, to and up, which are usually
% not capitalized unless they are the first or last word of the title.
% Linebreaks \\ can be used within to get better formatting as desired.
% Do not put math or special symbols in the title.
\title{A Discrete-time Reputation-based Resilient Consensus Algorithm for Synchronous or Asynchronous Communications\thanks{{\scriptsize 
This
work was supported in part by FCT project POCI-01-0145-FEDER-031411-HARMONY; partially supported by the project MYRG2018-00198-FST of the University of Macau; by the Portuguese Fundação para a Ciência e a Tecnologia (FCT) through Institute for Systems and Robotics (ISR), under Laboratory for Robotics and Engineering Systems (LARSyS) project UIDB/50009/2020, through project PCIF/MPG/0156/2019 FirePuma and through COPELABS, University Lusófona project UIDB/04111/2020.}}}
%
%
% author names and IEEE memberships
% note positions of commas and nonbreaking spaces ( ~ ) LaTeX will not break
% a structure at a ~ so this keeps an author's name from being broken across
% two lines.
% use \thanks{} to gain access to the first footnote area
% a separate \thanks must be used for each paragraph as LaTeX2e's \thanks
% was not built to handle multiple paragraphs
%

\author{Guilherme Ramos\thanks{{\scriptsize G. Ramos is with the Department of Electrical and Computer Engineering, Faculty of Engineering, University of Porto, Portugal. 
	%He acknowledges the support of Institute for Systems and Robotics, Instituto Superior T\'ecnico (Portugal), through scholarship BL112/2019\_IST-ID. 
	{\tt\small gramos@fe.up.pt}}}, 
	Daniel Silvestre\thanks{{\scriptsize D. Silvestre is with NOVA University of Lisbon and also with the Institute for Systems and Robotics, Instituto Superior T\'ecnico, University of Lisbon, 1049-001 Lisbon, Portugal. {\tt\small dsilvestre@isr.ist.utl.pt}}}, Carlos Silvestre % <-this % stops a space
\thanks{{\scriptsize C. Silvestre is with the Dep. of Electrical and Computer Engineering of the Faculty of Science and Technology of the University of Macau, Macau, China, on leave from Instituto Superior T\'ecnico/Technical University of Lisbon, Lisbon, Portugal. {\tt\small csilvestre@umac.mo}}}
}

\maketitle

% As a general rule, do not put math, special symbols or citations
% in the abstract or keywords.
\begin{abstract}
We tackle the problem of a set of agents achieving resilient consensus in the presence of attacked agents. We present a discrete-time reputation-based consensus algorithm for synchronous and asynchronous networks by developing a local strategy where, at each time, each agent assigns a reputation (between zero and one) to each neighbor. The reputation is then used to weigh the neighbors’ values in the update of its state.
Under mild assumptions, we show that: (i) the proposed method converges exponentially to the consensus of the regular agents; (ii) if a regular agent identifies a neighbor as an attacked node, then it is indeed an attacked node; (iii) if the consensus value of the normal nodes differs from that of any of the attacked nodes’ values, then the reputation that a regular agent assigns to the attacked neighbors goes to zero.
Further, we extend our method to achieve resilience in the scenarios where there are noisy nodes, dynamic networks and stochastic node selection. Finally, we illustrate our algorithm with several examples, and we delineate some attacking scenarios that can be dealt by the current proposal but not by the state-of-the-art approaches. 

\end{abstract}

% Note that keywords are not normally used for peerreview papers.
%{\scriptsize
% \begin{keywords}
% Agents and autonomous systems, Fault tolerant systems, Reputation systems, Computer Networks.
% \end{keywords}
% %}

% For peer review papers, you can put extra information on the cover
% page as needed:
% \ifCLASSOPTIONpeerreview
% \begin{center} \bfseries EDICS Category: 3-BBND \end{center}
% \fi
%
% For peerreview papers, this IEEEtran command inserts a page break and
% creates the second title. It will be ignored for other modes.
%\IEEEpeerreviewmaketitle

%\maketitle

%\section{Introduction}
% The very first letter is a 2 line initial drop letter followed
% by the rest of the first word in caps.
% 
% form to use if the first word consists of a single letter:
% \IEEEPARstart{A}{demo} file is ....
% 
% form to use if you need the single drop letter followed by
% normal text (unknown if ever used by the IEEE):
% \IEEEPARstart{A}{}demo file is ....
% 
% Some journals put the first two words in caps:
% \IEEEPARstart{T}{his demo} file is ....
% 
% Here we have the typical use of a "T" for an initial drop letter
% and "HIS" in caps to complete the first word.
\section{Introduction}\label{sec:intro}

In the last decades, much work has been devoted to networked control systems with a focus on cybersecurity aspects. 
Communications over shared mediums potentiates the exploitation of vulnerabilities that may result in consequences, sometimes, catastrophic. 

A central problem in networked multi-agent systems is the so-called consensus problem, where a set of agents interacting locally through a communication network attempt to reach a common value. In other words, the final value is the outcome of running a distributed algorithm among nodes which can communicate according to the network topology.    
Therefore, the problem of consensus is the basilar stone of a multitude of applications, ranging areas such as: optimization~\cite{tsitsiklis1986distributed,johansson2008subgradient}; motion coordination tasks -- flocking and leader following~\cite{jadbabaie2003coordination}; rendezvous problems~\cite{cortes2006robust}; computer networks resource allocation~\cite{chiang2007layering}; and the computation of the relative importance of web pages, by PageRank like algorithm~\cite{silvestre2018pagerank}.

Notwithstanding, the consensus problem further appears as a critical subproblem of major applications. 
 A distributed Kalman Filter based on two consensus systems was proposed in~\cite{olfati2007distributed} to estimate the 2D motion of a target. The work was experimentally assessed in~\cite{alriksson2007experimental} to estimate the motion of a real robot. 
Due to the relevance of consensus methods and the cybersecurity aspects, a crucial property to ensure in any consensus algorithm is its ability to overcome abnormal situations, i.e., achieve resilient consensus.

\textbf{Resilient consensus.}
Each agent in a network that works as expected must be able to filter erroneous information and be capable of reaching the consensus value that results from legit network information. 
A new fault-tolerant algorithm to accomplish approximate Byzantine consensus in asynchronous networks is introduced in~\cite{haseltalab2015approximate}. 
%The proposed method uses a topological condition for its success, less restrictive than that in previous results. 
The method also applies to synchronous networks, to networks with communication paths with delay, and the authors extend the results to time-varying networks.
Subsequently, authors performed the convergence rate analysis of the fault-tolerant consensus algorithm of~\cite{haseltalab2015approximate}
 in~\cite{haseltalab2015convergence}.

%A related approach to what we proposed in this paper consists in considering an overlay to the consensus algorithm that performs the detection and ensures resilience.
The work in \cite{silvestre2013SSVO} introduced a system to tackle worst-case and stochastic faults in the particular scenario of gossip consensus. 
The developed method may be integrated into the consensus algorithm to achieve resilient consensus, making the nodes converge to a steady state belonging to the set resulting from the intersection of the estimates that each agent perceives  for the remaining nodes~\cite{silvestre2014finite}. 
The technique was extended to a wider family of gossip algorithms in \cite{silvestre2017automatica}. 
These methods converge, and they offer a theoretical bound on the attacker signal, that can be computed \emph{a priori}. However,  their computational complexity for the worst-case undetectable attack makes the approach inviable.
For detecting attacks, their computational complexity is, in general, worse than that of the new algorithm proposed in this work. For detecting attacks, their computational complexity is, in general, worse than that of the new algorithm proposed in this work. For attackers' isolation, the computational complexity in~\cite{silvestre2017automatica,ramossilvestresCDC,ramossilvestresIJoC} is exponential, contrasting with the proposed polynomial-time approach. % we devise in this manuscript. 

In~\cite{oksuz2018distributed}, two fault-tolerant and parameter-independent consensus algorithms are introduced to deal with misbehaving agents. 
One of the approaches adaptively estimates, using a fault-detection scheme, how many faulty agents are in the network. 
Whenever there are $f$ faulty agents the authors' method converges if the network of non-faulty nodes is $(f + 1)$-robust. 
The other approach consists of a non-parametric Mean-Subsequence-Reduced algorithm, converging when both the network of non-faulty agents is $(f + 1)$-robust and the non-faulty agents possess the same amount of in-neighbors.

In~\cite{sundaram2018distributed}, the authors characterize the limitations on the performance of any distributed optimization algorithm in the presence of adversaries. Additionally, they investigated the vulnerabilities of consensus-based distributed optimization protocols to identify agents that are not following the consensus update rule. 
The authors suggest a resilient distributed optimization method, and they provide lower bounds on the distance-to-optimality of attainable solutions under any algorithm, resorting to the notion of maximal $f$-local sets of graphs for cases where each agent has at most $f$ adversaries neighbors.

In~\cite{dibaji2015consensus}, the authors explore the scenario where each regular agent in a network refreshes its state based on local information, using a deliberated feedback law, and that malicious nodes update their state arbitrarily. 
The authors present an algorithm for the consensus of second-order sampled-data multi-agent systems. 
Assuming that the network has sufficient connectivity, the authors proposed a resilient consensus method, in which each node ignores the information of neighbors having large/small position values, considering that although the global topology of the network is unknown to regular agents, they know a priory the maximum number of malicious neighbor nodes. 
Next, using similar reasoning, the authors in~\cite{kikuya2017fault} extend the previous work, presenting a consensus algorithm for clock synchronization in wireless sensor networks.

Subsequently, in~\cite{sundaram2016ignoring}, each agent may have a particular threshold for the number of extreme neighbors it ignores. 
Under such a setup, the author draws network conditions that ensure reaching a consensus. 
Further,  the author presented conditions under which the consensus is asymptotically almost surely (with probability 1) achieved in random networks and with random node's thresholds.  
 In~\cite{usevitch2018resilient}, it is proposed a resilient leader-follower consensus to arbitrary reference values, where each agent ignores a portion of extreme neighbors. This consensus method guarantees a steady state value lying in the convex hull of initial agent states. 
In a similar approach, the authors of~\cite{saldana2017resilient} presented a resilient consensus algorithm for time-varying networks of dynamic agents. 
In~\cite{dibaji2017resilient}, the case of quantized transmissions, with communication delays and asynchronous update schemes is studied for update times selected in either a deterministic or random fashion.

In~\cite{chen2018attack}, the authors present a consensus+innovations estimator. 
In the proposed method, each node in the network thresholds the gain to its innovations term. The authors ensure that if less than half of the agents’ sensors fall under attack, then all of the agents’ estimates converge with polynomial rate to the parameter of interest.
In~\cite{dibaji2019resilient}, the authors present a resilient fully distributed averaging algorithm that uses a resilient retrieval procedure, where all non-Byzantine nodes send their initial values and retrieve those of other agents. The convergence is ensured under a more restrictive than the conventional node connectivities assumption.

\textbf{Classes of consensus problems.}  
We may classify the problems of consensus based on the domain of the time update as: discrete-time (discrete), as in~\cite{zhu2010discrete,cai2012average}; or as continuous-time (continuous), as in~\cite{ren2005consensus,cai2012average}.
Also, it can be classified based on the network communication:  synchronous, see for instance~\cite{zhu2010discrete}; or asynchronous, see for example~\cite{tsitsiklis1986distributed,haseltalab2015approximate,liu2017asynchronous}. 
Moreover, the communication between agents may be: deterministic, see~\cite{cai2012average} for example; or stochastic, as in~\cite{boyd2006randomized,antunes2011consensus,silvestre2018broadcast} for instance. 
Lastly, the agents' network of communication can be categorized as: static, see~\cite{cai2012average,oksuz2018distributed,silvestre2018broadcast}; or dynamic, changing in time, see  in~\cite{olfati2004consensus,zhu2010discrete,saldana2017resilient}.  

Here, we propose a resilient consensus algorithm that can be used for discrete-time, synchronous or asynchronous communication and static or dynamic network of communication. 
By synchronous we mean algorithms in which updates occur at the same time for all the nodes, whereas in an asynchronous setting, at each time, any non-empty subset of nodes may update their state.  

\textbf{Reputation systems.}
The concept of reputation of an entity is an opinion about that entity that usually results from an evaluation of the entity based on a set of criteria. 
Everyday, we implicitly assign a reputation to persons, companies, services, and many other entities. 
We do so by evaluating the entity behavior and comparing to what we would expect, and, further, by assessing other important (with large reputation) entities opinion. 
In fact, reputation is an ubiquitous, spontaneous and highly efficient utensil of social control in natural societies, emerging in  business, education and online communities.

Hence, this ubiquitous concept has been ported to several important application with relative success, such as in ranking systems~\cite{li2012robust,saude2017robust,saude2017reputation,10.1145/3397271.3401278}, where reputation-based ranking systems are proposed and shown to cope better with attacks and to the effect of bribing. 
It is also an important measure to assess in the field of Social Networks as in~\cite{josang2007survey,zhu2015authenticated}. 
In~\cite{pujol2002extracting}, the authors propose a method to address the problem of calculating a degree of reputation for agents acting as assistants to the members of an electronic community.

For more related work on reputation systems, we refer the reader to the surveys  in~\cite{josang2007survey,hendrikx2015reputation} and references therein.

%Fully distributed
%Low computational power - averages per each neighbor
%

\textbf{Main contributions:}  
(i) we propose a (fully distributed) discrete-time, reputation-based consensus algorithm resilient to attacks that works for both synchronous and asynchronous networks; with this algorithm, each agent only needs to have a low computational power to do calculations with the neighbors' values;
 (ii) we show that the proposed algorithm converges with exponential rate; 
 (iii) for attacks with some properties, we show, theoretically, that our method does not produce false positives and holds polynomial time complexity, although we observe the same behavior for other types of attacks.

%\textbf{Paper structure}. The remainder of this paper is organized as follows. 
%In Section~\ref{sec:notation}, we set the adopted notation in the manuscript. 
%Section~\ref{sec:rep_based} details the problem we aim to address and presents a reputation-based consensus algorithm (\textsf{RepC}) to solve it. 
%%In Section~\ref{sub:complexity}, we analyze the computational complexity of our algorithm. 
%In Section~\ref{sec:ill_exam}, we provide illustrative examples of the proposed \textsf{RepC} algorithm for multitude of attacking scenarios. 
%Lastly, Section~\ref{sec:conc} closes the paper and draws avenues for future research. 
%  
\subsection{Preliminaries and Terminology}\label{sec:notation}
First, we recall some concepts of graph theory, and we set the notation that we will adopt in this manuscript. 
A \emph{directed graph}, or simply a \emph{digraph}, is an ordered pair $\mathcal G=(\mathcal V,\mathcal E)$, where $\mathcal V$ is a set of $n>1$ \emph{nodes}, and $\mathcal E\subseteq \mathcal V\times\mathcal V$ is a set of \emph{edges}. Edges are ordered pairs which represent a relationship of accessibility between nodes. If $u,v\in\mathcal V$ and $(u,v)\in\mathcal E$ then the node $v$ directly accesses information of node $u$. 	
In the scope of consensus algorithms, we also refer to a digraph as a \emph{network}, and we further say that nodes are \emph{agents} of the network.  
A \emph{complete digraph} or \emph{complete network} is a digraph such that all nodes can directly access information of every other node. 
Given an agent $v\in\mathcal V$, we denote the set of nodes that $v$ can directly access information in the network $\mathcal G=(\mathcal V, \mathcal E)$ by $\mathcal N_v=\{v\}\cup\{u\,:\,(u,v)\in\mathcal E\}$, and they are the \emph{neighbors} of $v$.  The proper neighbors of $v$ are $\bar{\mathcal N_v}=\mathcal N_v\setminus \{v\}$. 
The \emph{in-degree} of a node $v\in\mathcal V$, denoted by $d_v$ is the number of neighbors of $v$, i.e.,  
$d_v=|\mathcal N_v|$. Likewise, the \emph{out-degree} of a node $v\in\mathcal V$, $o_v$, is the number of nodes that have $v$ has neighbor, i.e., $o_v=|\{u\,:\,v\in\mathcal N_u\}|$. 
A \emph{path} in $\mathcal G=(\mathcal V,\mathcal E)$ is a sequence of nodes $(v_1,v_2,\hdots v_k)$ such that $(v_i,v_{i+1})\in\mathcal E$ for every $i=1,\hdots, k-1$. 
%{\color{red}Sera que se pode colocar a reputacao como o peso da aresta, facilita a exposicao?} 
%A directed graph is \emph{strongly connected} whenever for all $u,v\in\mathcal V$ there is a path from node $u$ to node $v$, i.e., all node can access information of all nodes. 
%A \emph{weighted digraph} is a digraph $\mathcal G=(\mathcal V,\mathcal E,\mathcal W)$, where each edge has a weight $w\in\mathbb R$, i.e., the result of the weight function $\mathcal W:\mathcal E\to \mathbb R$. 
A convenient way of representing a digraph is by means of its adjacency matrix $A\in\mathbb R^{n\times n}$, where $A_{u,v}=1$ if $(u,v)\in\mathcal E$, and $A_{u,v}=0$, otherwise. 
%In the weighted case, $A_{u,v}=w$ if $(u,v)\in\mathcal E$ and $\mathcal W(u,v)=w$, and $A_{u,v}=0$, otherwise. 
A \emph{subgraph} or a \emph{subnetwork} $\mathcal H = (\mathcal V',\mathcal E')$ of a digraph $\mathcal G=(\mathcal V,\mathcal E)$ is a digraph such that $\mathcal V'\subset \mathcal V$, $\mathcal E'\subset\mathcal E$. 
% and the restriction of $\mathcal W$ to the edges in $\mathcal E'$ coincides to $\mathcal W'$, i.e., $\mathcal W_{|\mathcal E'}=\mathcal W$. 
If $\mathcal A$ denotes a set of nodes $\mathcal A\subset\mathcal V$, we denote by $\mathcal G\setminus\mathcal A$ the subgraph $\mathcal H$ of $\mathcal G$ that consists in $\mathcal H=(\mathcal V\setminus\mathcal A,\mathcal E')$, where $\mathcal E'=\{e\in\mathcal E\,:\,e=(u,v)\text{ and } u,v\notin\mathcal A\}$.

Given a finite and non-empty array of possibly repeated elements sorted by increasing order  $\mathcal C=\{x_1,\ldots, x_n\}$, with $x_i\in\mathbb R$, with at least one $j\neq i$ and $x_i\neq x_j$, we define the following element:
$
	\displaystyle\fmin_{x\in\mathcal C}x = y_f,
$
where, inductively defined as  
\[
 y_f = {\small\begin{cases}
		x_1, & \text{if }f=1\\[0.2cm]
		 \min\left\{x\in\mathcal C\,:\,y_{f-1} < x < x_n\right\}, & \text{if }\displaystyle\mathop{\exists}_{x\in\mathcal C}y_{f-1} < x < x_n \\
		 		y_{f-1}, & 
\text{otherwise.}
	\end{cases}}
\]

%{\color{red}EXEMPLO}
In other words, we are computing the $f$th smallest element of the set obtained from the array $\mathcal C$ by discarding repeated elements, but ensuring that it is not the maximum element. 
This definition will be very important to the reputation-based consensus method we propose, because we need to normalize a set of values,  dividing them by the difference between the maximum and the $\fmin$ element.  
 For instance, consider the sets of sorted elements $\mathcal C=\{1,2,2,2,3\}$. 
If $f=1$ then $\displaystyle\fmin_{x\in \mathcal C}x=1$, and if $f\geq 2$ then $\displaystyle\fmin_{x\in \mathcal C}x=2$.
%Given a vector with size $k$, $w=[w_1,w_2,\hdots,w_k]$, we denote the shift of each element of the vector by one position by $Shift(w)$, i.e. $Shift(w)=[w_2,\hdots,w_k,w_1]$.
%We denote the norm-1 of the vector $v$ with the standard notation $\|v\|_1=\sum_{i=1}^k |v_i|$. 
%Given a set $\mathcal S\subset\mathbb N$, we denote the collection of its subsets by $\wp(\mathcal S)$, identifying all permutations of sorted elements. For instance, if $\mathcal S=\{1,2,3\}$, then $\wp(\mathcal S)=\{ \{\},\{1\},\{2\},\{3\},\{1,2\},\{1,3\},\{2,3\},\{1,2,3\} \}$.
%Analogously, we define the subsets with size $i\leq |\mathcal S|$ as $\wp(\mathcal S,i)=\{w\in\wp(\mathcal S)\,:\,|w|=i\}$.

%We use the standard logical quantifiers, $\forall$ and $\exists$, and we denote by $\exists!_x.\varphi(x)$ as the abbreviation $\exists!_x.\varphi(x)\equiv\exists_x.\varphi(x)\land\forall_{y\neq x}.\lnot\varphi(y)$, i.e. it denotes the exists one and only one quantifier.

\section{Reputation-Based Consensus}\label{sec:rep_based} 

%{\color{red}
Consider a network of agents $\mathcal G=(\mathcal V,\mathcal E^{(k)})$, with initial states $x_v^{(0)}\in\mathbb R$, for $v\in \mathcal V$. 
In the non-attacked scenario, agents can reach consensus through the use of a distributed linear iterative algorithm with dynamics given by: 
\begin{equation}\label{eq:dyn}
x^{(k+1)} = W^{(k)} x^{(k)},
\end{equation}
where $x^{(k)}$ is the vector collecting the $n$ agents states at time step $k>0$, and the matrix $W^{(k)}\in\mathbb R^{n\times n}$ is such that: (i) $W_{u,v}=0$ if the agents $u$ and $v$ do not communicate, and (ii) the agents converge to the same quantity, i.e.,  $\displaystyle\lim_{k\to\infty}x^{(k)}=x^\infty$.  
%}

Additionally, we consider a set of attacked agents $\mathcal A\subset\mathcal V$. 
If the agents $a\in\mathcal A$ do not follow the update rule of the consensus procedure, then each regular agent, $v\in\mathcal V\setminus\mathcal A$, should be able to identify and discard the attacked agents' values in the computation of the consensus value. In order to solve the problem stated%in~\ref{sub:ps}
, we make the following assumption:

%\begin{itemize}
%	\item 
\noindent$\,\,\bullet\,\,$	For each regular agent,  $v\in\mathcal V\setminus\mathcal A$, more than half of the neighbors are regular agents, i.e., $|\mathcal N_v\cap\mathcal A|<|\mathcal N_v|/2$ and the network of normal nodes is connected.
%\end{itemize}

\noindent The assumption we made is a typical assumption in the state-of-the-art methods to achieve resilient consensus.
We remark that the assumption we do make is equivalent to the $r$-robustness ($(r,1)$-robust) defined in~\cite{kikuya2017fault}.

Further, observe that the previous assumption is reasonable, because each regular agent needs to be able to divide his neighbors into the set of normal nodes and the set of attacked ones, by comparing the information that it receives. Hence, if the majority of the information is not legitimate there is no redundancy to allow to identify the attacked neighbors. 

%{\color{red}
\subsection{Attacker model}
%{\color{blue}
In what follows, we consider an attacker that may corrupt the state of the nodes in the subset $\mathcal{A}$ by adding an unbounded signal. The attacked dynamics are a corrupted version of~\eqref{eq:dyn} as follows:
\setcounter{equation}{1}
\begin{equation}
% \begin{cases}
% x^{(1)} \quad\,= W x^{(0)},&\\
 x^{(k+1)} = W^{(k)} x^{(k)} + \Delta^{(k)},
 %&\text{if }k>0
% \end{cases}
\end{equation}
where $\Delta^{(k)}\in\mathbb R^n$, which entails the assumption that the attacker cannot corrupt the communication between nodes to send different messages to distinct neighbors. 
Observe that this assumption allows the attacker to change the state of a subset of agents to (possibly) different values. For example, in a network with 10 agents, an attacker may change the state of agents 3 and 5 using different values, but it cannot change the network communication scheme.  

Furthermore, the attacker cannot create artificial nodes nor change the network topology, i.e., the structure of $W^{(k)}$ and the dimension $n$ are fixed in (2). 
Notice that if a malicious entity could create nodes in the network, it would be impossible to deter, as the attacked nodes would be the majority regardless of $n$. 

Additionally, the attack cannot target the initial state, i.e.,  $\Delta^{(0)} = 0$, since this scenario would be undetectable. The sequences of state values for the attacked version would be the same as a normal execution of the algorithm with the attack value as initial state. 
%}
%we are only interested in attacking scenarios that can be detected.

% {\color{red}
% (ISTO ANTES DE A)
% \[
% x^{(k+1)} = W x^{(k)},
% \]

% \[
% x^{(k+1)} = W x^{(k)} + \Delta^{(k)},
% \]
% Relacionar com o texto acima e dizer o que diz respeito ao que
% }

%Thus, it is a common assumption. 
%}

\subsection{Reputation-based consensus \emph{(\textsf{RepC})}}\label{sec:rbc} 

Next, we propose a reputation-based  consensus algorithm (\textsf{RepC}). 
The idea behind the algorithm is that, each time an agent obtains information (states) from its neighbors, the agent measures how discrepant is, in average, the state from one neighbor regarding the states of the remaining ones and its own state. 
The \textsf{RepC} is composed by two phases: (i) \textbf{identification of the attacked nodes}; (ii) computation of the \textbf{consensus}. 
 
%{\color{red}
Notice that the proposed algorithm is a fully distributed discrete-time consensus algorithm  that works for both synchronous and asynchronous networks. Also, each agent only needs to have a low computational power to do calculations with the neighbors' values.
%}

\subsubsection{Synchronous communication \emph{\textsf{RepC}}}
Given the maximum number of allowed attacked nodes $f$, the identification of the attacked nodes is performed by the following iterative scheme:
{\small
%\begin{equation}\label{eq:id} 
\allowdisplaybreaks
\begin{align*}\label{eq:id} 
     &\textbf{Reputation update: }\\
		\tilde c_{ij}^{(k+1)} &  =  \begin{cases} 
			1-
	\displaystyle\displaystyle\sum_{v\in \bar{\mathcal N_i} }\frac{|x_j^{(k)}-x_v^{(k)}|}{|\mathcal N_i|},	& j\in\mathcal N_i\\
			0, &\text{otherwise}
		\end{cases}\\		     &\textbf{Normalized Reputation update: }
			\\[0.2cm]
		 \tilde\tilde c_{ij}^{(k+1)} & =  \begin{cases} \displaystyle\frac{\tilde c_{ij}^{(k+1)}-\displaystyle\fmin_{v\in\bar{\mathcal N_i}}\tilde c_{iv}^{(k+1)}}{\displaystyle\max_{v\in\bar{\mathcal N_i}}\tilde c_{iv}^{(k+1)}-\displaystyle \fmin_{v\in\bar{\mathcal N_i}}\tilde c_{iv}^{(k+1)}}, & i\neq j\\
		 1, & \text{otherwise} 
 \end{cases}\\
      & \textbf{Normalized Reputation update with confidence $\varepsilon$: }
\\
		 c_{ij}^{(k+1)} & =  \begin{cases} 
		 \,\tilde\tilde c_{ij}^{(k+1)}, &\text{if }\,\tilde\tilde c_{ij}^{(k+1)}>0, \\
		 \varepsilon^{k+1},& \text{otherwise}\\
		\end{cases}\\
		     & \textbf{Consensus state update: }
        \\ 
		x_i^{(k+1)} & = \displaystyle\sum_{j\in\bar{\mathcal N_i}} c_{ij}^{(k)}x_j^{(k)}\bigg/\displaystyle\sum_{j\in\bar{\mathcal N_i}} c_{ij}^{(k)},\\
\end{align*}
%\end{equation}
}
{\small $c_{ij}^{(k+1)}=0$} if {\small $j\notin\mathcal N_i$}, and {\small $c_{ii}^{(k+1)}=1$}, and
where  {\small $c_{ij}^{(0)}=1$} for {\small $j\in\bar{\mathcal N_i}$} and {\small $c_{ij}^{(0)}=0$} otherwise, and {\small $x_i^{(0)}$} is the initial value of each agent $i$.
Further, {\small $\varepsilon\in]0,1[$} is a \emph{confidence factor} which guarantees that each agent does not discard immediately values that are discrepant from its neighbors' average. 

%{\color{red}
Notice that the selected value for $\varepsilon$ must be small to have a negligible impact on the agents'  consensus states. Also, a large $\varepsilon$ may cause an agent to do not detect an attacked neighbor. This, in turn, makes the asymptotic consensus to deviate from the consensus without attacked agents towards a combinations of the attacked agents asymptotic states. 
We illustrate this property in Section~\ref{sec:ill_exam}. %} 

%{\color{red} 
Further, notice that the proposed method computes a weighted average of the agents' values. Therefore, we can ensure that the final consensus state is a convex combination of the agents' initial states. 
%}

\subsubsection{Asynchronous communication \emph{\textsf{RepC}}}
The asynchronous version of algorithm \textsf{RepC} consists of, at each instance of time, the agents that communicate, $\mathcal A'\subset\mathcal A$, follow Equation~\eqref{eq:id}, where $\bar{\mathcal  N}_i$ is replaced by $\bar{\mathcal N}_i\cap \mathcal A'$.

Another interesting fact is that the iterative scheme~\eqref{eq:id} may also be used in the scenario where the network of agents evolves with time. 
The results in Section~\ref{sec:rbc} can be restated for this scenario by considering that the set of neighbors of a node is dynamic, and by verifying, at each time, that each agent has more than two neighbors and more than half of them are regular agents. 
In Section~\ref{sub:dyn_net} and~\ref{sub:dyn_net_noisy}, we illustrate the dynamic network of agents and dynamic network with noisy agents scenarios.

The first important property to prove about \textsf{RepC} is that it converges. 
To simplify the proof, we assume that we are in the scenario of synchronous communication. 
The general proof follows the same steps, but it is more complex and it needs more complex notation to denote the set of neighbors with which a node communicates at each time. 
Additionally, we assume that $x_i^{(k)}\in[0,1]$. 
Notice that this corresponds to do a bijection of each $x_i^{(0)}$ as $\tilde x_{i}^{(0)}=( x_{i}^{(0)}-x_{\min}^{(0)})\big/(x_{\max}^{(0)}-x_{\min}^{(0)})$. 
Further, in the following proofs and for technical reasons, we assume that each attacked agent shares a state that is converging to some value. 
This simplification is needed in order to derive theoretical guarantees of the proposed method. 
Although, in practice, the algorithm is still effective under other circumstances, as we illustrate in Section~\ref{sec:ill_exam}.

\begin{lemma}\label{lemma:conv}
	If for any $i\in\mathcal V$ we have that $|\mathcal N_i|>2$, then each agent that follows the iterative scheme in~\eqref{eq:id} converges.
\end{lemma}

\begin{proof}
	$
	%\begin{split}
		\left\|x^{(k+1)}-x^{(k)}\right\|_\infty = \displaystyle\max_{i}\left|x_i^{(k+1)}-x_i^{(k)}\right| %\\
	%\end{split},
	$		
		 and, hence, assuming without loss of generality that $\|c_i^{(k)}\|_1\leq\|c_i^{(k+1)}\|_1$	
% 	\begin{equation}\label{eq:proof1}
	{\small
	\begin{align*}\label{eq:proof1}
		\left|x_i^{(k+1)}-x_i^{(k)}\right|  = \left|\frac{c_i^{(k+1)}\cdot x^{(k)}}{\|c_i^{(k+1)}\|_1}-\frac{c_i^{(k)}\cdot x^{(k-1)}}{\|c_i^{(k)}\|_1}\right| \\
		 = \left| \frac{c_i^{(k+1)}\cdot x^{(k)}}{\|c_i^{(k+1)}\|_1}-\frac{c_i^{(k)}\cdot x^{(k)}}{\|c_i^{(k+1)}\|_1}+\frac{c_i^{(k)}\cdot x^{(k)}}{\|c_i^{(k+1)}\|_1}
		-\frac{c_i^{(k)}\cdot x^{(k-1)}}{\|c_i^{(k)}\|_1} \right|\\
%	\end{split}
%	}
%	\]
%	Assuming, without loss of generality, that $\|c_i^{(k)}\|_1\leq\|c_i^{(k+1)}\|_1$,
%	% we have that	
%		\[
%		{\small
%		\begin{split}
%		 \left| \frac{c_i^{(k+1)}\cdot x^{(k)}}{\|c_i^{(k+1)}\|_1}-\frac{c_i^{(k)}\cdot x^{(k)}}{\|c_i^{(k+1)}\|_1}+\frac{c_i^{(k)}\cdot x^{(k)}}{\|c_i^{(k+1)}\|_1}
%		-\frac{c_i^{(k)}\cdot x^{(k-1)}}{\|c_i^{(k)}\|_1} \right|\\
		\leq \left| \frac{c_i^{(k+1)}\cdot x^{(k)}}{\|c_i^{(k+1)}\|_1}-\frac{c_i^{(k)}\cdot x^{(k)}}{\|c_i^{(k+1)}\|_1}+\frac{c_i^{(k)}\cdot x^{(k)}}{\|c_i^{(k)}\|_1}
		-\frac{c_i^{(k)}\cdot x^{(k-1)}}{\|c_i^{(k)}\|_1} \right|\\
		= \left| \frac{c_i^{(k+1)}-c_i^{(k)}}{\|c_i^{(k+1)}\|_1}\cdot x^{(k)}+c_i^{(k)}\cdot\frac{x^{(k)}-x^{(k-1)} }{\|c_i^{(k)}\|_1} \right|\\
%	\end{split}
%	}
%	\]
%	Because we are assuming that $x_v^{(m)},c_{ij}^{(m)}\in]0,1[$, then we have that 
%	{\small\begin{equation}\label{eq:proof1}
%	\begin{split}
%	\left| \frac{c_i^{(k+1)}-c_i^{(k)}}{\|c_i^{(k+1)}\|_1}\cdot x^{(k)}+c_i^{(k)}\cdot\frac{x^{(k)}-x^{(k-1)} }{\|c_i^{(k)}\|_1} \right|\\
	\leq \frac{\max_{j\in\mathcal N_i}|c_{ij}^{(k+1)}-c_{ij}^{(k)}|}{\|c_i^{(k+1)}\|_1}+\frac{\max_{j\in\mathcal N_i}|x_{j}^{(k)}-x_{j}^{(k-1)}|}{\|c_i^{(k)}\|_1}\\
	\leq \frac{\max_{j\in\mathcal N_i}|c_{ij}^{(k+1)}-c_{ij}^{(k)}|}{\|c_i^{(k)}\|_1}+\frac{\max_{j\in\mathcal N_i}|x_{j}^{(k)}-x_{j}^{(k-1)}|}{\|c_i^{(k)}\|_1},
	\end{align*} 
	}
% 	\end{equation}
	because we are assuming that $x_v^{(m)},c_{ij}^{(m)}\in]0,1[$. 
	Now we need to compute $\max_{j\in\mathcal N_i}|c_{ij}^{(k+1)}-c_{ij}^{(k)}|$.
	First, we notice that we cannot have that $\max_{j\in\mathcal N_i}|c_{ij}^{(k+1)}-c_{ij}^{(k)}|=|\varepsilon^{k+1}-\varepsilon^{k}|$, because there is always a $j\in\mathcal N_i$ such that $c_{ij}^{(k+1)}>\varepsilon^{k+1}$ and all the other $k\neq j\in\mathcal N_i$ are such that $c_{ij}^{(k+1)}\geq\varepsilon^{k+1}$.
	Therefore, we need to consider only three cases:
{\small	\begin{enumerate}
		\item $c_{ij}^{(k+1)}=\varepsilon^{k+1}$ and $c_{ij}^{(k)}=\tilde\tilde c_{ij}^{(k)}$;
		\item $c_{ij}^{(k+1)}=\tilde\tilde c_{ij}^{(k+1)}$ and $c_{ij}^{(k)}=\varepsilon^k$;
		\item $c_{ij}^{(k+1)}=\tilde\tilde c_{ij}^{(k+1)}$ and $c_{ij}^{(k)}=\tilde\tilde c_{ij}^{(k)}$.
	\end{enumerate}}
	For case $1)$ we have that 
	{\small$
		\left|c_{ij}^{(k+1)}-c_{ij}^{(k)}\right| = \left|\varepsilon^{k+1} - \tilde\tilde c_{ij}^{(k)} \right| 
		 < \left|\,\tilde\tilde c_{ij}^{(k+1)} - \tilde\tilde c_{ij}^{(k)} \right|,
	$}
	since {\small $c_{ij}^{(k+1)}=\varepsilon^{k+1}$} implies that {\small $\tilde\tilde c_{ij}^{(k+1)}\leq 0$}.
	Using the same reasoning, for case $2)$, we have that 	
	{\small$
		\left|c_{ij}^{(k+1)}-c_{ij}^{(k)}\right| = \left|\, \tilde\tilde c_{ij}^{(k+1)}-\varepsilon^{k} \right| 
		 < \left|\,\tilde\tilde c_{ij}^{(k+1)} - \tilde\tilde c_{ij}^{(k)} \right|.
	$}
	We only need to compute $3)$
%	\[
%	{\small
%	\begin{split}
{\small	\begin{equation}\label{eq:proof2}
	\begin{split}	
		\left|\,\tilde\tilde c_{ij}^{(k+1)}-\tilde\tilde c_{ij}^{(k)}\right| = \frac{\left|\tilde c_{ij}^{(k+1)}-\tilde c_{ij}^{(k)}\right|}{\max_{v\in\bar{\mathcal N_i}}\tilde c_{iv}^{(k+1)}-\fmin_{v\in\bar{\mathcal N_i}}\tilde c_{iv}^{(k+1)}}\\
		\leq |\tilde c_{ij}^{(k+1)}-\tilde c_{ij}^{(k)}|\\
		= \frac{1}{|\bar{\mathcal N_i}|}\sum_{v\in \bar{\mathcal N_i} }\left(|x_j^{(k)}-x_v^{(k)}|-|x_j^{(k-1)}-x_v^{(k-1)}|\right)\\
%	\end{split}
%	}
%	\]
%	Let $\alpha=\displaystyle\underset{v\in\bar{\mathcal N_i}}{\arg\max} \left(|x_j^{(k)}-x_v^{(k)}|-|x_j^{(k-1)}-x_v^{(k-1)}|\right)$, then 
%{\small	\begin{equation}\label{eq:proof2}
%		\begin{split}	
%		 \displaystyle\frac{1}{|\mathcal N_i|}\displaystyle\sum_{v\in \bar{\mathcal N_i} }\left(|x_j^{(k)}-x_v^{(k)}|-|x_j^{(k-1)}-x_v^{(k-1)}|\right)\\
		\leq  \frac{1}{|\mathcal N_i|}|\mathcal N_i|\left(|x_j^{(k)}-x_\alpha^{(k)}|-|x_j^{(k-1)}-x_\alpha^{(k-1)}|\right)\\
		=   \left(|x_j^{(k)}-x_\alpha^{(k)}| -|x_{\alpha}^{(k-1)}-x_{j}^{(k-1)}|\right.\\
		 +\left.|x_{\alpha}^{(k-1)}-x_{j}^{(k-1)}|  -|x_j^{(k-1)}-x_\alpha^{(k-1)}|\right)\\
		\leq  \left(|x_j^{(k)}-x_j^{(k-1)}|+|x_\alpha^{(k)}-x_\alpha^{(k-1)}|\right)
		%\\
		 \leq   2 \displaystyle\max_{j\in\bar{\mathcal N_i} }|x_j^{(k)}-x_{j}^{(k-1)}|,
	\end{split}  
	\end{equation}}
	
	\noindent where $\alpha=\displaystyle\underset{v\in\bar{\mathcal N_i}}{\arg\max} \left(|x_j^{(k)}-x_v^{(k)}|-|x_j^{(k-1)}-x_v^{(k-1)}|\right)$. 
	Now, plugging~\eqref{eq:proof2} in~\eqref{eq:proof1}, we have that
	\[
	{\small
	\begin{split}
		\|x^{(k+1)}-x^{(k)}\|_\infty \leq \frac{3}{\|c_i^{(k))}\|_1}\max_{j\in\bar{\mathcal N_i}}|x_j^{(k)}-x_{j}^{(k-1)}|\\
		\leq \frac{3}{\|c_i^{(k)}\|_1}\|x^{(k)}-x^{(k-1)}\|_\infty\leq \frac{3}{|\bar{\mathcal N_i}|}\|x^{(k)}-x^{(k-1)}\|_\infty
		\end{split}
		}
	\]
	Therefore, the iterative scheme converges whenever $\frac{3}{|\mathcal N_i|+1}<1$, which is equivalent to $|\mathcal N_i|>2$. 
\end{proof}

In fact, the requirement about the number of neighbors, i.e., each agent having more than 2 neighbors, agrees with the intuition. A regular agent should assess to, at least, 3 states to distinguish whether nodes are following the consensus update rule or not. 
Otherwise, if there are only 2 neighbors then their reputation can be, for instance, alternating between iterations. 

%{\color{red} Show that each regular agent converges to the same value $x^\infty$}

The previous result states that \textsf{RepC} converges (i.e. each agent converges to a state) but it is still missing to show that each regular agent converges to the same value, i.e., all regular agents  \emph{agree}.  
The next lemma assesses that: (i) either an agent $v$ converges to a unique value; (ii) or for any other agent, the reputation of agent $v$ is zero, i.e. $c_{uv}^\infty=0$.

\begin{lemma}\label{lemma:one_limit}
	Consider the iterative scheme~\ref{eq:id}.
	For any agent $j\in\mathcal V$ one of the following holds:
	\begin{enumerate}
		\item[$(i)$] $\displaystyle\lim_{k\to\infty}x_j^{(k)}=x^\infty$;
		\item[$(ii)$] $\displaystyle\forall_{i\in\mathcal N_j}\lim_{k\to\infty}c_{ij}^{(k)}=0$ (neighbors of $j$ assign it reputation zero).
	\end{enumerate}
\end{lemma}

\begin{proof}
	By Lemma~\ref{lemma:conv}, we have that~\eqref{eq:id} converges for each agent $i\in\mathcal V\setminus\mathcal A$.
	We just need to show that for a given node $u\in\mathcal V\setminus\mathcal A$ and for each of its neighbors $v\in\mathcal N_u$ either $(i)$ or $(ii)$ happens.
	Let $u\in\mathcal V\setminus\mathcal A$ and $v\in\mathcal N_u$, we show by induction on the number of neighbors of $u$, $|\mathcal N_u|$, that for each neighbor either its reputation is zero or it converges to the same value as $u$.
	The basis is when $|\mathcal N_u|=1$, and we have that $x^\infty_u=\frac{x^\infty_u+c^\infty_{uv}x_v^\infty}{1+c^\infty_{uv}}$. Thus, either $c^\infty_{uv}=0$, or $c^\infty_{uv}>0$ and $x_v^\infty=x_u^\infty$.
	When $|\mathcal N_u|=N+1$, we have that 
	$
%	\begin{split}
x^\infty_u  = \sum_{j\in\bar{\mathcal N}_u}c_{uj}^\infty x_j^\infty\bigg/\sum_{j\in\bar{\mathcal N}_u}c_{uj}^\infty.
$
Since the reputation that $u$ assigns to itselft is $c^\infty_{uu}=1$, we can rewrite the previous expression as
{\small\begin{equation}\label{eq:step1}
x^\infty_u 	 = \frac{x_u^\infty+\sum_{j\in\mathcal N_u\setminus\{v\}}c_{uj}^\infty x_j^\infty+c_{uv}^\infty x_v^\infty}{1+\sum_{j\in\mathcal N_u\setminus\{v\}}c_{uj}^\infty+c_{uv}^\infty}
%	\end{split}
\end{equation}}
	Further, using the induction hypothesis, either $(i)$ or $(ii)$ is true for any set of $N$ neighbors of $u$. Hence, for $j\in\mathcal N_u\setminus\{v\}$ either $x_j^\infty=x^\infty$ or $c_{uj}^\infty=0$. In any of the cases, we have that   
{\small	\begin{equation}\label{eq:step2}
x_u^\infty+\sum_{j\in\mathcal N_u\setminus\{v\}}c_{uj}^\infty x_j^\infty = x^\infty\left(1+\sum_{j\in\mathcal N_u\setminus\{v\}}c_{uj}^\infty\right)
\end{equation}
}
	By replacing~\eqref{eq:step2} in~\eqref{eq:step1}, it follows that
		\[{\small
	\begin{split}x^\infty_u = x^\infty & = \frac{x^\infty\left(1+\sum_{j\in\mathcal N_u\setminus\{v\}}c_{uj}^\infty\right)+c_{uv}^\infty x_v^\infty}{\left(1+\sum_{j\in\mathcal N_u\setminus\{v\}}c_{uj}^\infty\right)+c_{uv}^\infty},\\
	\end{split}
	}
\]
implying that either $c_{uv}^\infty=0$ or  $c_{uv}^\infty>0$ and $x_v^\infty=x_u^\infty=x^\infty$.
	By transitivity, we can apply the same to each neighbor of all neighbors of $u$, and so forth. Thus, the result yields for all $i\in\mathcal V\setminus\mathcal A$.
\end{proof}

As a corollary, we have that if a regular agent using \textsf{RepC} detects a neighbor as a faulty node, then it is a faulty node. 
In other words, there are no false positives.

%{\color{red} Show that $x^\infty\neq x^a$}

\begin{corollary}\label{cor:detection}
		Let $v\in\mathcal V\setminus\mathcal A$ and $u\in\mathcal V$. 
	 	By using the iterative scheme~\eqref{eq:id}, if $c_{uv}^\infty=0$ then $u\in\mathcal A$.
\end{corollary}

The proof of Lemma~\ref{lemma:conv} also hints that half of each agent's neighbors should not be under attack, so that each normal node identifies the attacked agents correctly.  
This is expressed in the following.

\begin{lemma}\label{lemma:rep}
	Suppose that the iterative scheme~\eqref{eq:id} converges to a value different from that broadcasted by the attacked agents. If for each agent $i\in\mathcal V\setminus\mathcal A$, less than half of its neighbors are not attacked agents, i.e. $|\bar{\mathcal N_i}\cap \mathcal A|<|\bar{\mathcal N_i}\setminus \mathcal A|$, then $\displaystyle\lim_{k\to\infty}c_{ia}^{(k)}=0$, for $a\in\mathcal A$ and $\displaystyle\lim_{k\to\infty}c_{iv}^{(k)}=1$ for $v\in\bar{\mathcal N_i}\setminus\mathcal A$.
\end{lemma} 

\begin{proof}
	By Lemma~\ref{lemma:conv}, we have that the each regular agent using the iterative scheme in~\eqref{eq:id} converges to $x^\infty$. Let $y$ denote the value that all the attacked agents in $\mathcal A$ share with the neighbors.
	Thus, for a regular agent $i\notin\mathcal A$, an attacked agent's reputation, $a\in\mathcal A$, satisfies 
	\[
	{\small
	\begin{split}
		\tilde c_{ia}^\infty=&\lim_{k\to\infty}\tilde c_{ia}^{(k)} = 1-\frac{1}{|\mathcal N_i|}\sum_{v\in\mathcal N_i}\left|y-\lim_{k\to\infty}x_v^{(k)}\right|
		\\= &
		 1-\frac{|\mathcal N_i\setminus\mathcal A|}{|\mathcal N_i|}\left|y-x^{\infty}\right|,
		\end{split}
		}
	\] 
	and the limit of the reputation of a regular user, $j\notin\mathcal A$, is given as
	{\small\[ 
%	\begin{split}
		\tilde c_{ij}^\infty= \lim_{k\to\infty}\tilde c_{ij}^{(k)}
		 =  
%		 1-\frac{1}{|\mathcal N_i|}\sum_{v\in\mathcal N_i}\left|\lim_{k\to\infty} \left(x_j^{(k)}-x_v^{(k)}\right)\right|
%		\\= &
		 1-\frac{|\mathcal N_i\cap\mathcal A|}{|\mathcal N_i|}\left|x^{\infty}-y\right|.
%		\end{split}
	\]	}
	Since $|\mathcal N_i\cap\mathcal A|<|\mathcal N_i\setminus\mathcal A|$ and $y\neq x^{\infty}$, then $\tilde c_{ij}^\infty>\tilde c_{ia}^\infty$, and because reputations values are normalized to be between $0$ and $1$, we have that, for all $i$, $1=c_{ij}^\infty>c_{ia}^\infty=0$. 
\end{proof}

Now, we need to show that a regular agent using \textsf{RepC}  identifies the neighbors which are attacked nodes, and to study the convergence rate of method.

\begin{lemma}\label{lemma:complete}
		Let $v\in\mathcal V\setminus\mathcal A$ and $u\in\mathcal V$. 
	 	By using the iterative scheme~\eqref{eq:id}, if $u\in\mathcal A$ and $|\mathcal A|\leq f$ then $c_{uv}^\infty=0$.
\end{lemma}

\begin{proof}
	Let $a\in\mathcal A$ be an attacked node.
	We want to show that for a regular agent, $v\in\mathcal V\setminus\mathcal A$, the reputation of agent $a$ strictly decreases with time.
	Let $v\in\mathcal V\setminus\mathcal A$, we have that
	$
	%\begin{split}
	|x_a-x_v^{(k+1)}|-|x_a-x_v^{(k)}|%&
	 \leq|x_v^{(k+1)}-x_v^{(k)}|.%\\
	%\end{split}
	$
\end{proof}

%Lastly, we study the convergence rate of \textsf{RepC}.

\begin{proposition}
	Consider the iterative scheme in~\eqref{eq:id} and let $N=\displaystyle\min_{i\in\mathcal V}|\mathcal N_i|$ and $\lambda = \frac{3}{N+1}$. If $N>3$, then~\eqref{eq:id} converges with exponential rate and we have that $\|x^{(k+1)}-x^{(k)}\|_\infty\leq \lambda^k$. Further, to achieve an error of at most $\delta>0$ between the last two iterations, we need to run the iterative scheme at most $k=\lceil\log_{\lambda}(\delta)\rceil$ times.
\end{proposition}
	
\begin{proof}		
	Let $N=\displaystyle\min_{i\in\mathcal V}|\mathcal N_i|$ and $\lambda = \frac{3}{N+1}$.
	Using the proof of Lemma~\ref{lemma:conv}, we have that 
	$	
		\|x^{(k+1)}-x^{(k)}\|_\infty \leq  \frac{3}{|\bar{\mathcal N_i}|}\|x^{(k)}-x^{(k-1)}\|_\infty
		 \leq \lambda^k\|x^{(1)}-x^{(0)}\|_\infty\leq \lambda^k.
	$
	Hence, the iterative scheme converges with exponential rate of $\lambda^k$.
	To achieve an error between iterations of at most $\varepsilon$, we need to have that 
	$
		\lambda^k\leq\delta,
	$
	which is equivalent to have that
	$
		k\leq\log_{\lambda}(\delta)\leq \lceil\log_{\lambda}(\delta)\rceil.
	$
	Therefore, if we run the iterative scheme at most $\lceil\log_{\lambda}(\delta)\rceil$ times, we obtain an error between the last two iterations of at most $\varepsilon$.
\end{proof}

%Now that we know that \textsf{RepC} works as expected, we study its computational complexity.

\subsection{Complexity Analysis}\label{sub:complexity}

Next, we investigate the complexity analysis of the proposed algorithm \textsf{RepC}, when the network communication is synchronous.
 
\begin{proposition}\label{prop:complexity}
	Let $\mathcal G=(\mathcal V,\mathcal E )$ be a network of agents, $l=\displaystyle\max_{v\in\mathcal V}|\mathcal N_v|$, then, for $i$ iterations and for each agent, the iterative scheme~\eqref{eq:id} has time complexity of $\mathcal O(l^2i)$.
\end{proposition}

\begin{proof}
	Given a network of agents $\mathcal G=(\mathcal V,\mathcal E )$, for time step $k$ and for an agent $v$, the time complexity of~\eqref{eq:id} is the sum of the time complexities of computing $\tilde c_{vu}^{(k)}$,  $\tilde\tilde c_{vu}^{(k)}$, $ c_{vu}^{(k)}$, for each $u\in\mathcal N_i$, and $x_v^{(k)}$.
	Computing $\tilde c_{vu}^{(k)}$ has computation complexity of $\mathcal O(|\mathcal N_v|^2)$, because there are $\mathcal O(|\mathcal N_v|^2)$ pairs of neighbors values to compute the absolute difference.
	Each of the remaining steps has time complexity of $\mathcal O(|\mathcal N_v|)$.
	Hence, the sum of each step time complexity is $O(|\mathcal N_v|^2)+4\times\mathcal O(|\mathcal N_v|)=O(|\mathcal N_v|^2)$.
	Thus, if $l=\displaystyle\max_{v\in\mathcal V}|\mathcal N_v|$, then $\mathcal O(l^2)$ is a bound for the time complexity that each incurs.
	Therefore, for $i$ iterations of the iterative scheme~\eqref{eq:id}, each agent incurs in $\mathcal O(l^2 i)$ time complexity.	
\end{proof}

%\section{Resilience against noisy agents}\label{sec:noisy} 
%
%Now, we extend the previous algorithm to achieve resilience agains noisy agents, i.e., agents that share a value accordingly to a probability distribution with determined expected value.
%
%
%
%
%The set of results we obtained in Section~\ref{sec:rbc} can be also extended to this case, but in terms of expected value.
%We illustrate this scenario in Section~\ref{sub:dyn_net_noisy}.
%

%\section{Dynamic Network}\label{sec:dyn} 

\section{Illustrative Examples}\label{sec:ill_exam}

Subsequently, we illustrate the use of  \textsf{RepC} for different kinds of attacks. 
Further, in the examples, we use $\varepsilon=0.1$.

\subsection{Same value for attacked nodes}

In the following examples, we consider the network of agents depicted in Fig.~\ref{fig:GA}~(a).

\begin{figure}[h!] 
\centering      
\subfigure[Network of agents $\mathcal G_A$.]{\label{fig:a}\includegraphics[width=0.42\textwidth]{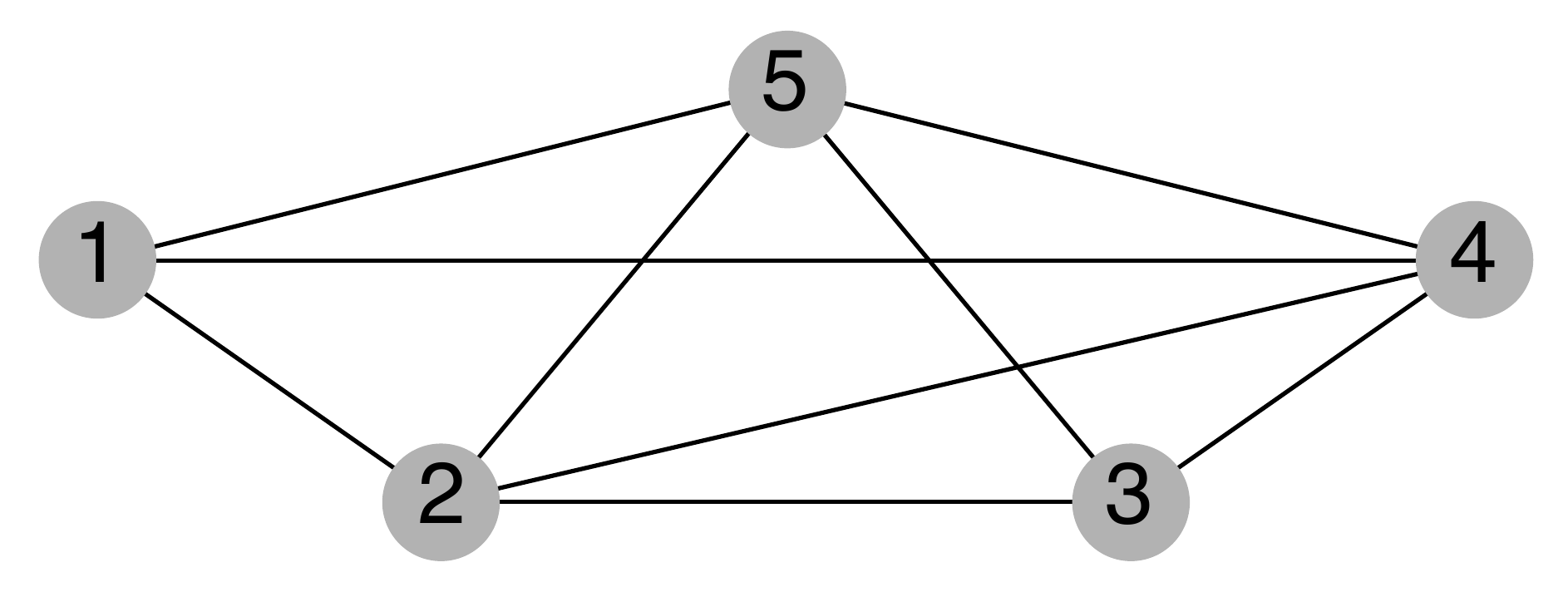}}
\subfigure[Network of agents $\mathcal G_B$.]{\label{fig:b}\includegraphics[width=0.38\textwidth]{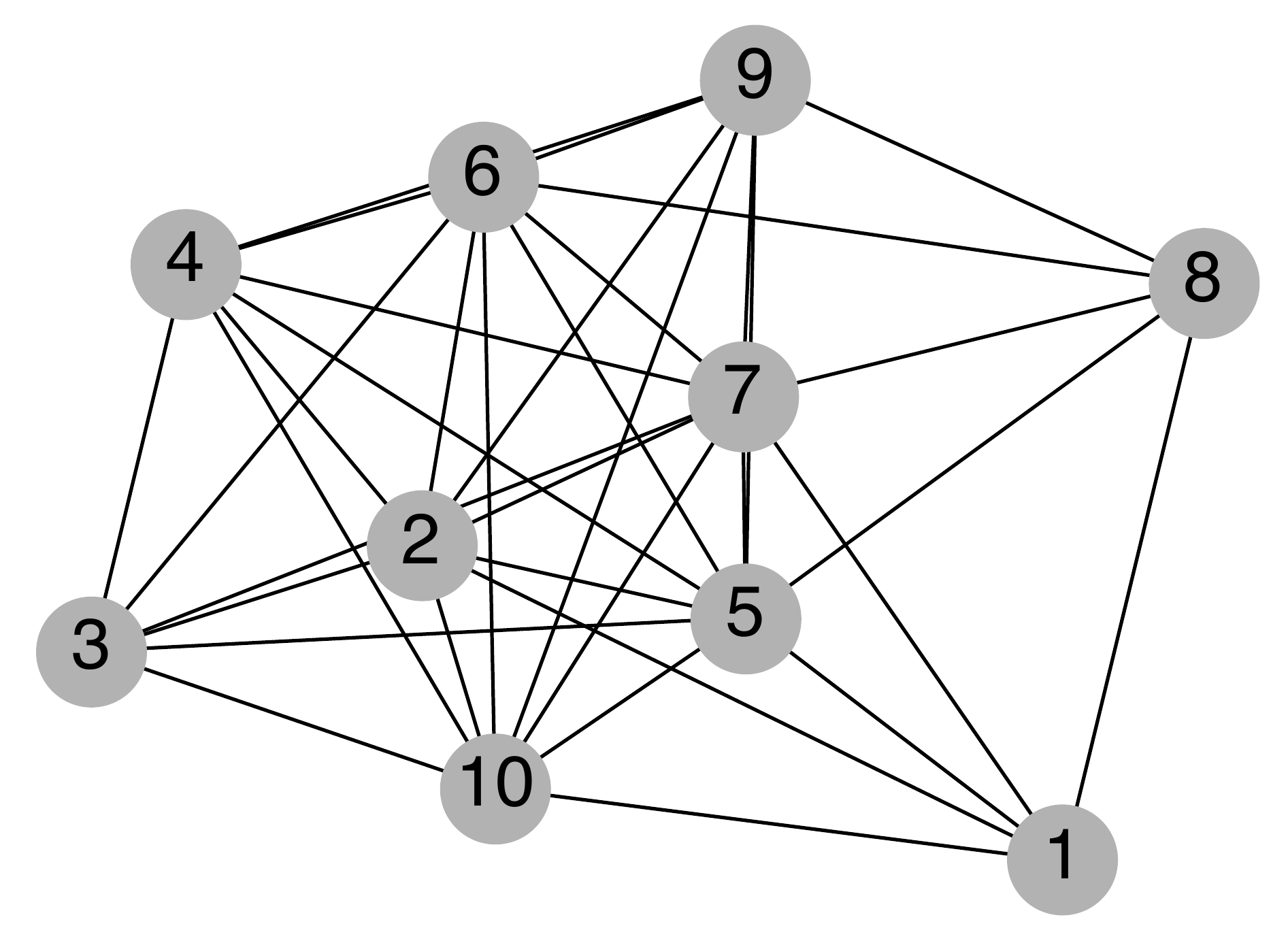}}
%\includegraphics[width=0.25\textwidth]{exDynG2.pdf}
%\caption{Network of agents $\mathcal G_D$.}
\caption{}
\label{fig:GA} 
\end{figure}

%\begin{figure}[h!] 
%\centering      
%\includegraphics[width=0.25\textwidth]{ex1G.pdf}
%\caption{Network of agents $\mathcal G_A$.}
%\label{fig:GA} 
%\end{figure}

%\subsubsection{Network without attacked nodes}

First, we illustrate algorithm \textsf{RepC} in the scenario of a network of agents \textbf{without attacked nodes}. The set of agents is $\mathcal V_1=\{1,\hdots,5\}$ and, thus, the set of attacked agents is $\mathcal A=\emptyset$. We set the parameter $f=1$. 
Figure~\ref{fig:C0} depicts the state evolution of each agent.%, and Figure~\ref{fig:c5_0} the evolution of the reputations agent $5$ assigns to each of its neighbors.  

\begin{figure}[h!] 
\centering      
\includegraphics[width=0.61\textwidth]{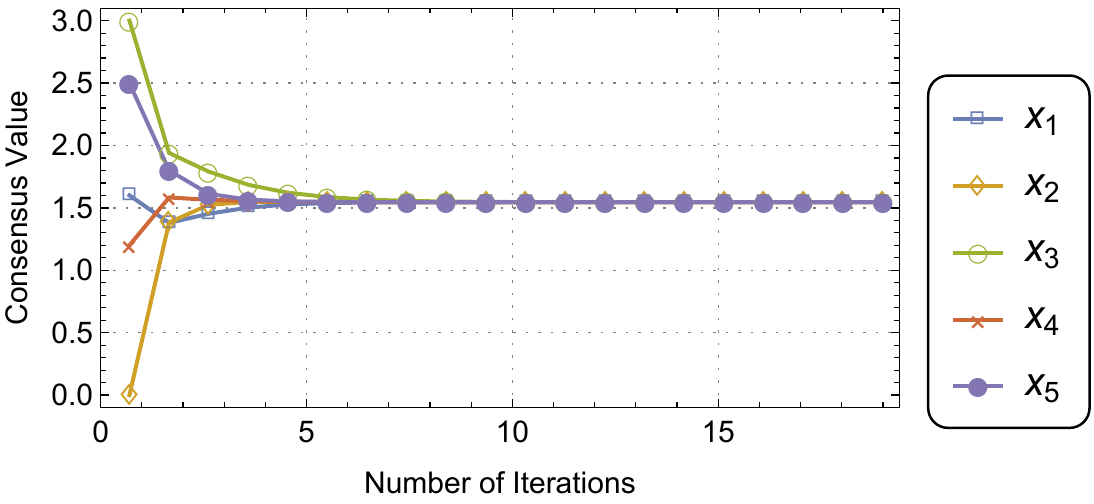}
\caption{Consensus of network $\mathcal G_A$ with agents $\mathcal V_1$ and set of attacked agents $\emptyset$.}
\label{fig:C0} 
\end{figure}

%\begin{figure}[h!] 
%\centering      
%\includegraphics[width=.49\textwidth]{c1_0.pdf}
%\caption{Evolution of the reputations that agent $1$ assigns to each of its neighbors.}
%\label{fig:c1_0} 
%\end{figure}

%\begin{figure}[h!] 
%\centering      
%\includegraphics[width=.49\textwidth]{c2_0.pdf}
%\caption{Evolution of the reputations that agent $2$ assigns to each of its neighbors.}
%\label{fig:c2_0} 
%\end{figure}
%
%\begin{figure}[h!] 
%\centering      
%\includegraphics[width=.49\textwidth]{c3_0.pdf}
%\caption{Evolution of the reputations that agent $3$ assigns to each of its neighbors.}
%\label{fig:c3_0} 
%\end{figure}
%
%\begin{figure}[h!] 
%\centering      
%\includegraphics[width=.49\textwidth]{c4_0.pdf}
%\caption{Evolution of the reputations that agent $4$ assigns to each of its neighbors..}
%\label{fig:c4_0} 
%\end{figure}

%\begin{figure}[h!] 
%\centering      
%\includegraphics[width=.49\textwidth]{c5_0.pdf}
%\caption{Evolution of the reputations that agent $5$ assigns to each of its neighbors.}
%\label{fig:c5_0} 
%\end{figure}

%\subsubsection{One attacked node}

Here, we explore the scenario where an \textbf{attacker} targets one agent to \textbf{share a value close to the consensus} of the network of regular agents, depicted in Fig.~\ref{fig:GA}.

\begin{figure}[h!] 
\centering      
\includegraphics[width=0.61\textwidth]{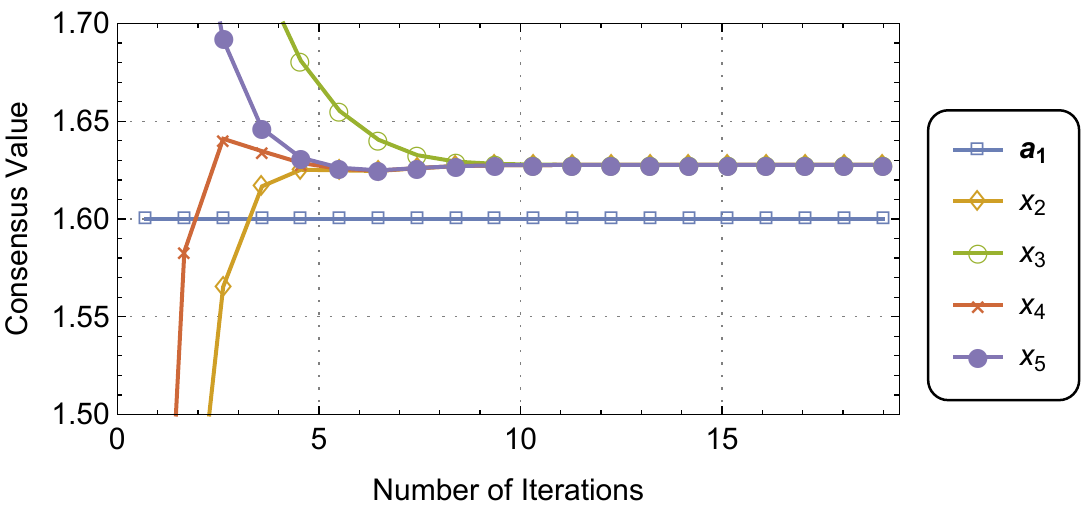}%{Exp1.pdf}
\caption{Consensus of network $\mathcal G_A$ with agents $\mathcal V_1$ and set of attacked agents $\mathcal A_1$ (zoomed around the attacker's value.)}
\label{fig:C1} 
\end{figure}

The set of agents is $\mathcal V=\{1,\ldots,5\}$ and the set of attacked agents is $\mathcal A=\{1\}$. 
Figure~\ref{fig:C1} depicts each agent consensus value. %, and Figure~\ref{fig:C1zoom} magnifies each agent state with a range close to consensus value. 
We can see that although the attacker value is very close to the consensus value, the neighbors of the attacked node assign zero to its reputation, by using~\eqref{eq:id}. Hence the value that the attacked node shares is discarded.
%\begin{figure}[h!] 
%\centering      
%\includegraphics[width=.49\textwidth]{Exp1zoom.pdf}
%\caption{Zoomed vision of Figure~\ref{fig:C1} around the attacker's value.}
%\label{fig:C1zoom} 
%\end{figure}

Next, in Figure~\ref{fig:c2} -- Fig.~\ref{fig:c3}, we depict the evolution of the reputations that agents $2$ and $3$ assign to their neighbors. 

\begin{figure}[h!] 
\centering      
\includegraphics[width=0.61\textwidth]{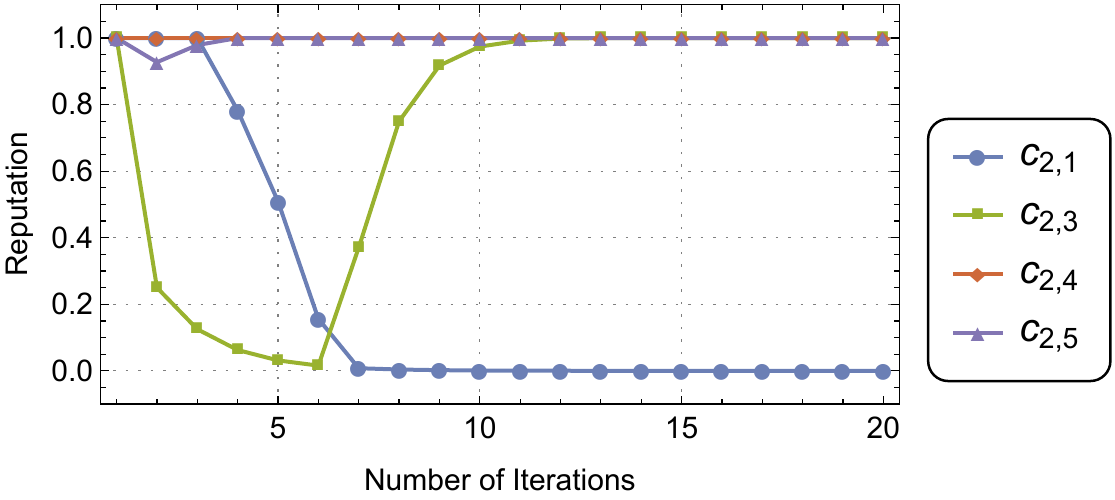}
\caption{Evolution of the reputations that agent $2$ assigns to each of its neighbors.}
\label{fig:c2} 
\end{figure}

\begin{figure}[h!] 
\centering      
\includegraphics[width=0.61\textwidth]{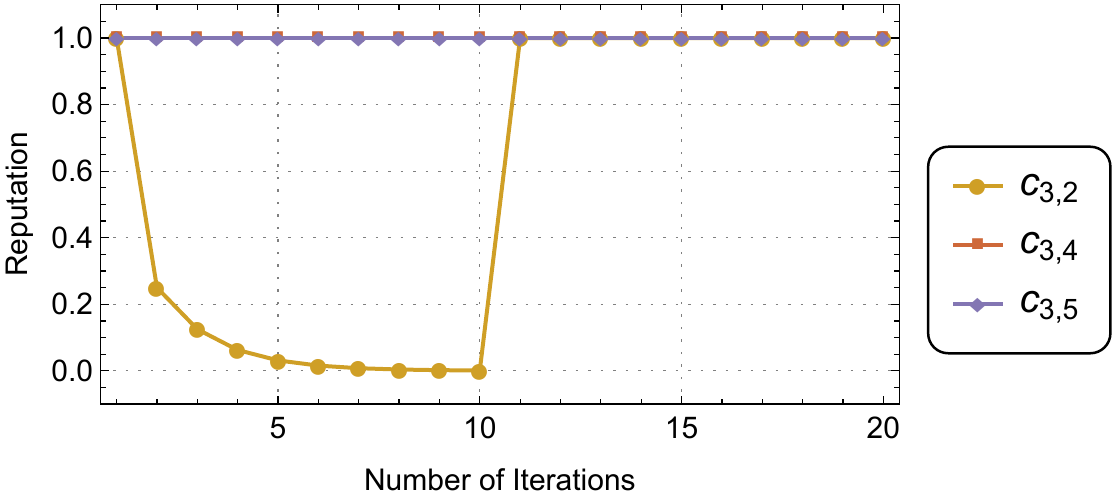}
\caption{Evolution of the reputations that agent $3$ assigns to each of its neighbors.}
\label{fig:c3} 
\end{figure}

%\begin{figure}[h!] 
%\centering      
%\includegraphics[width=.49\textwidth]{c4.pdf}
%\caption{Evolution of the reputations that agent $4$ assigns to each of its neighbors.}
%\label{fig:c4} 
%\end{figure}
%
%\begin{figure}[h!] 
%\centering      
%\includegraphics[width=.49\textwidth]{c5.pdf}
%\caption{Evolution of the reputations that agent $5$ assigns to each of its neighbors.}
%\label{fig:c5} 
%\end{figure}

%\vspace{-0.3cm}
\subsection{Different values for attacked nodes}

Next, we illustrate the scenario where attacked nodes share different values. 
For that end, we consider the set of agents $\mathcal V=\{1,\ldots,10\}$, with $\mathcal A=\{1,8\}$, and the network of agents depicted in Fig.~\ref{fig:GA}~(b).
We explore two scenarios with two attacked agents: (i) both attacked nodes share values (distinct) smaller than the consensus, see Fig.~\ref{fig:C2}; (ii) one attacked node shares a value larger than the consensus while the other uses a smaller value than the consensus, see Fig.~\ref{fig:C3}.
%\begin{figure}[h!] 
%\centering      
%\includegraphics[width=0.23\textwidth]{exDiffG1.pdf}
%\caption{Network of agents $\mathcal G_B$}
%\label{fig:GB} 
%\end{figure}

\begin{figure}[h!] 
\centering      
\includegraphics[width=0.61\textwidth]{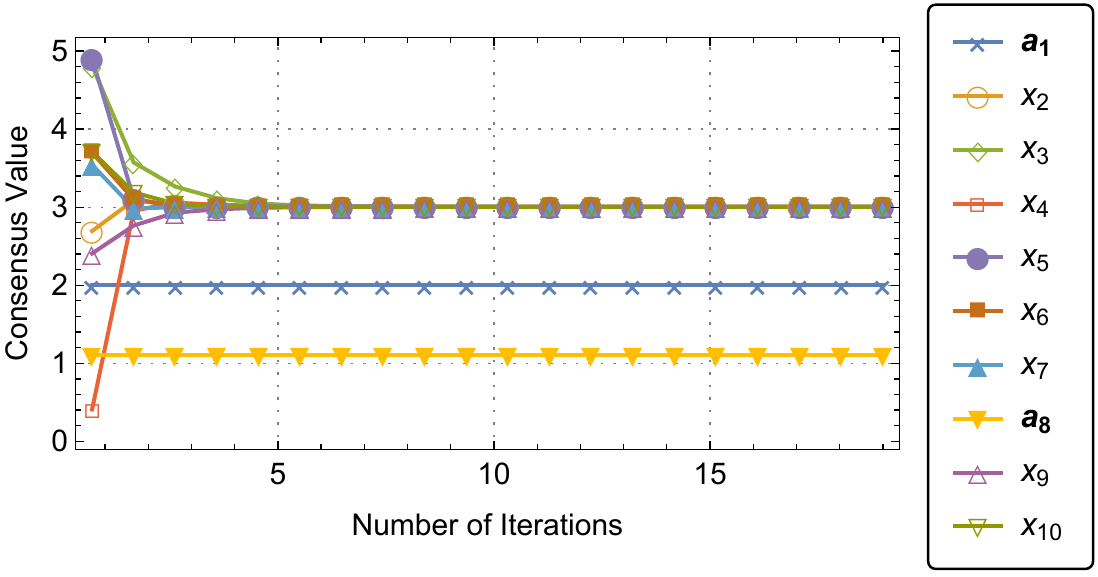}
\caption{Consensus of network $\mathcal G_B$ with agents $\mathcal V_3$ and set of attacked agents $\mathcal A_3$}
\label{fig:C2} 
\end{figure}
%\vspace{-0.45cm}
\begin{figure}[h!] 
\centering      
\includegraphics[width=0.61\textwidth]{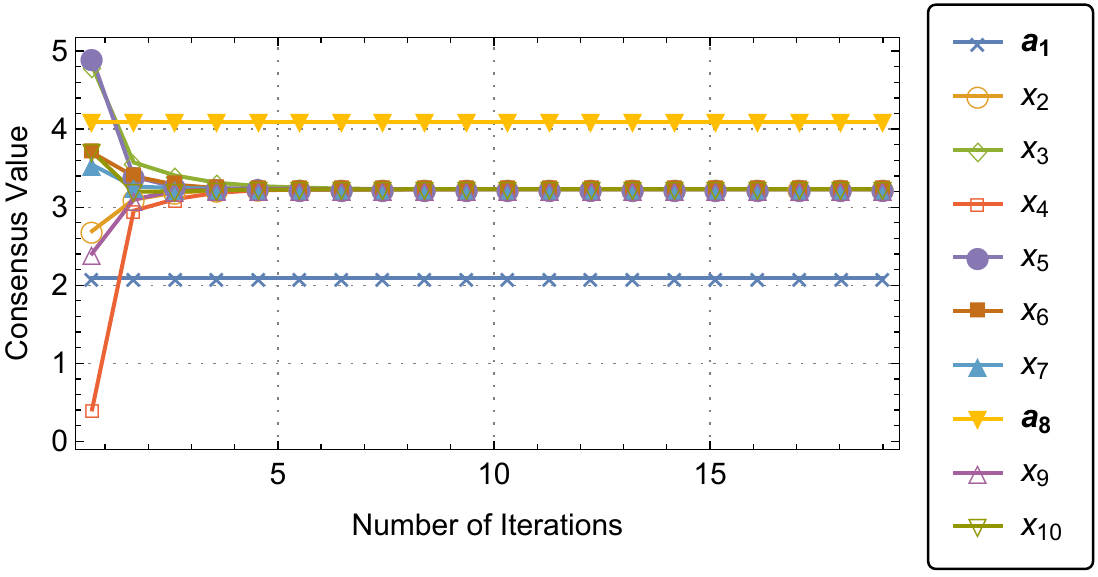}
\caption{Consensus of network $\mathcal G_B$ with agents $\mathcal V_2$ and set of attacked agents $\mathcal A_2$}
\label{fig:C3} 
\end{figure}

%\vspace{-0.45cm}
\subsection{Asynchronous Communication}
%\vspace{-0.1cm}

We, now, illustrate the use of algorithm \textsf{RepC} in the case where the communication between nodes occur asynchronously. To simulate this scenario, at each time instance, a random subset of agents communicates. 
The set of agents is $\mathcal V=\{1,\ldots,5\}$, the network of agents is $\mathcal G_A$, and the set of attacked agents is $\mathcal A=\{1\}$. 
Figure~\ref{fig:repConAsyn} depicts the state evolution of each agent when using the asynchronous version of algorithm \textsf{RepC}. Each normal node identifies and discards the information of the attacked agent.

\begin{figure}[h!] 
\centering      
\includegraphics[width=0.61\textwidth]{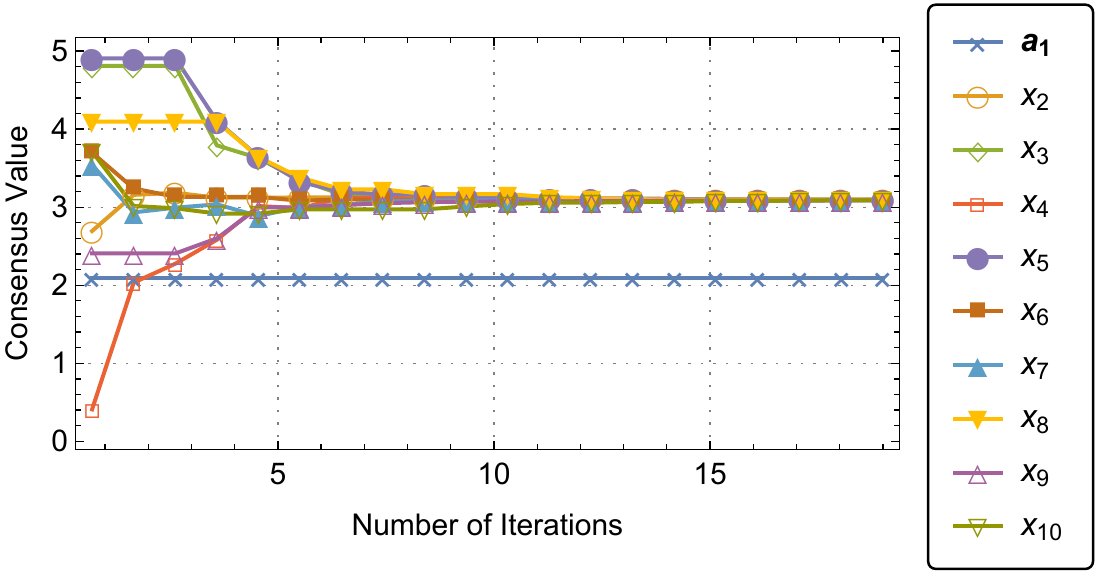}
\caption{Consensus of network $\mathcal G_A$ with agents $\mathcal V_1$, asynchronous communication, and set of attacked agents $\mathcal A=\{1\}$.}
\label{fig:repConAsyn} 
\end{figure}

%\vspace{-0.45cm}
\subsection{Dynamic network}\label{sub:dyn_net}
Next, we test the scenario where the network of agents evolves with time and the attacked agents share the same value.
We consider two networks composed by 10 agents, as depicted in Fig.~\ref{fig:GD}, with set of agents $\mathcal V=\{1,\hdots,10\}$ and set of attacked agents $\mathcal A=\{1\}$.
%\begin{figure}[h!] 
%\centering      
%\includegraphics[width=0.25\textwidth]{exDynG1.pdf}
%\caption{Network of agents $\mathcal G_C$.}
%\label{fig:GC}  
%\end{figure}

%\vspace{-0.2cm}
\begin{figure}[h!] 
\centering    
\hfill
\subfigure[Network of agents $\mathcal G_D$.]{\label{fig:a}\includegraphics[width=0.46\textwidth]{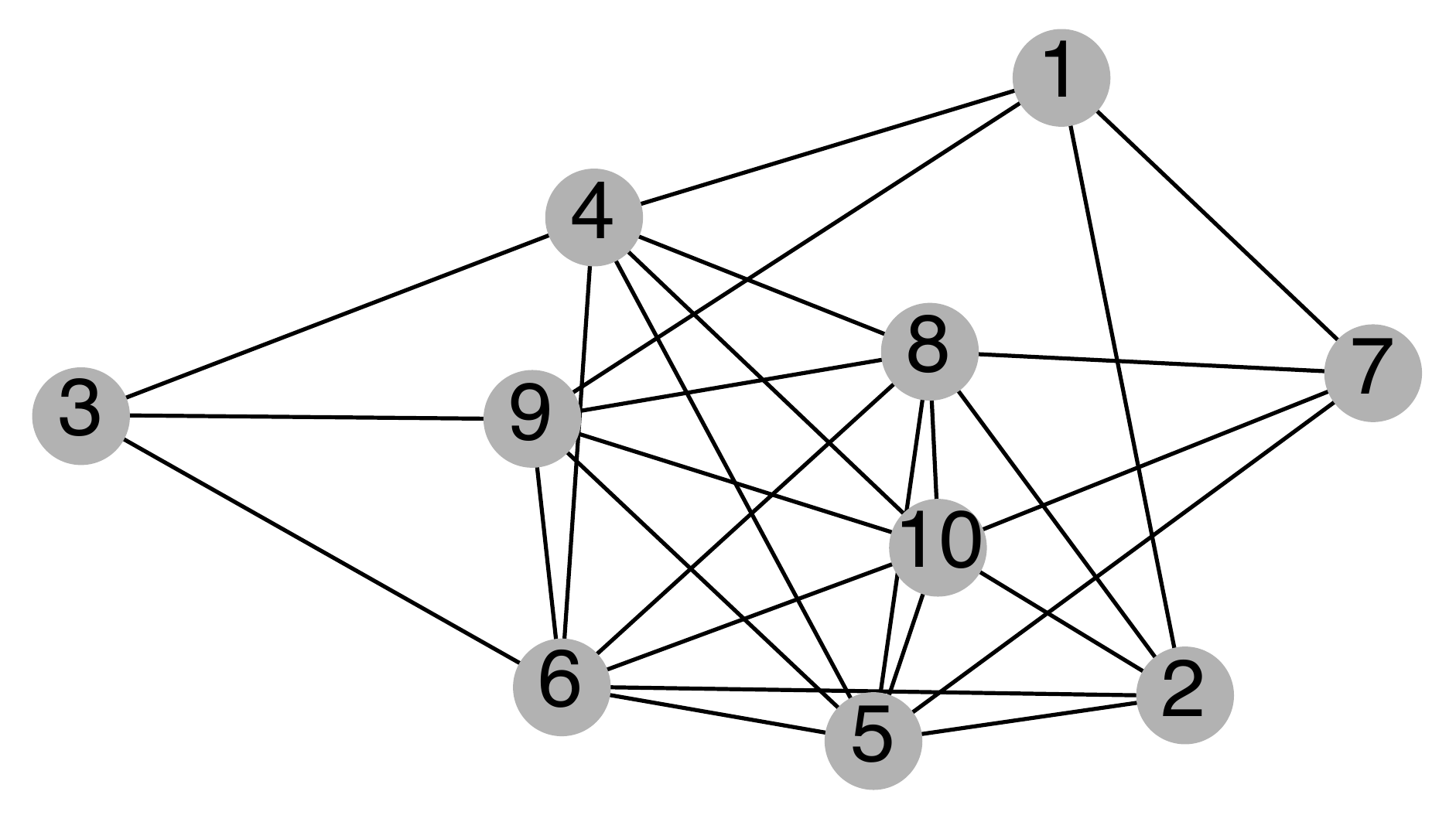}}\hfill
\subfigure[Network of agents $\mathcal G_C$.]{\label{fig:b}\includegraphics[width=0.46\textwidth]{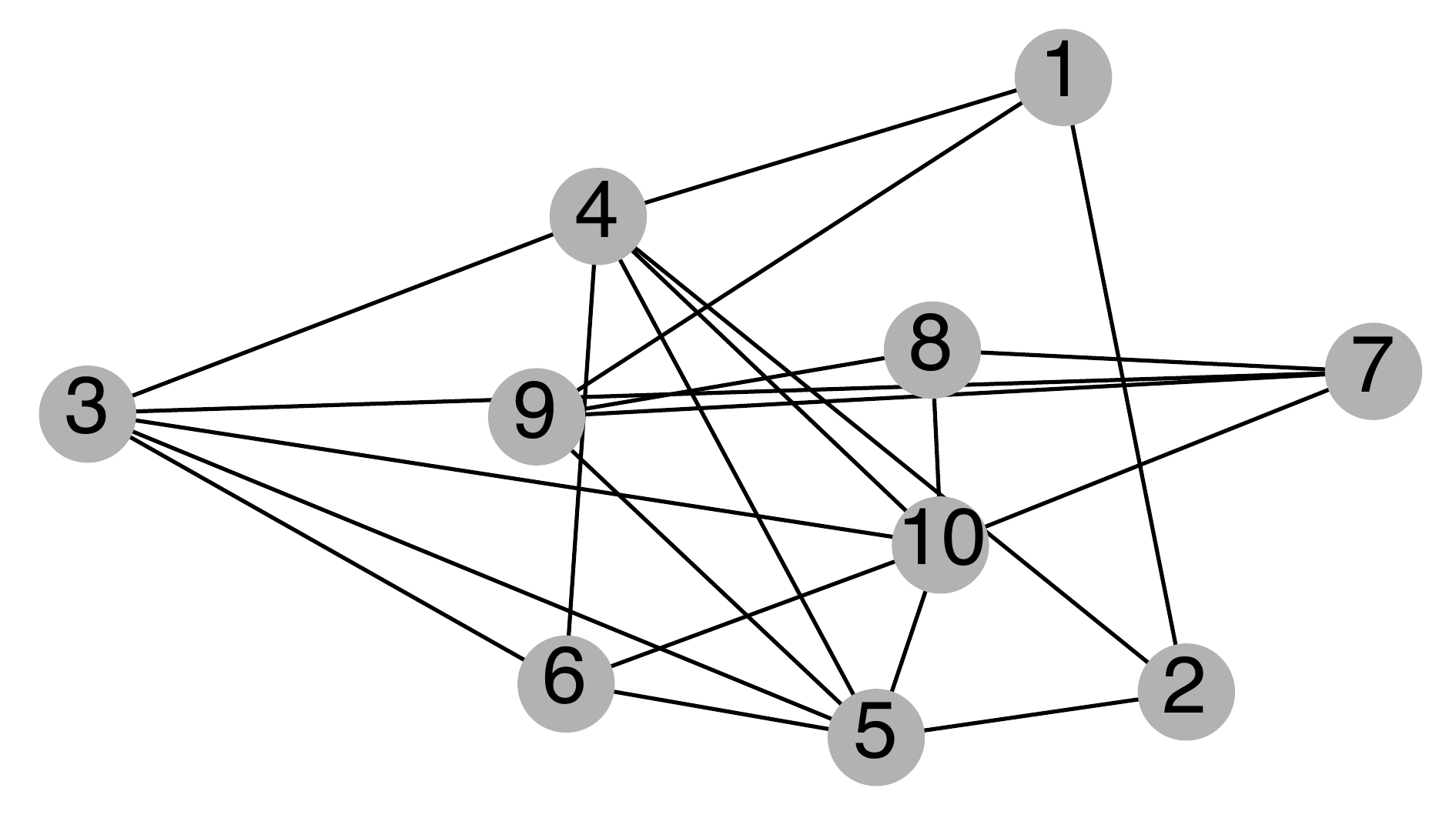}}
\hfill
\caption{}
\label{fig:GD} 
\end{figure}
%\vspace{-0.1cm}

In the example, we consider that the dynamic network of agents for time instance $k>0$ is given by 
$\mathcal G_1^{(k)}=\begin{cases}
	\mathcal G_C &\text{if } k\leq 10\\
	\mathcal G_D &\text{otherwise}
\end{cases}.$ 
The consensus value of each agent, utilizing the iterative scheme~\eqref{eq:id}, is depicted in Fig.~\ref{fig:C4}.
\begin{figure}[h!] 
\centering      
\includegraphics[width=0.61\textwidth]{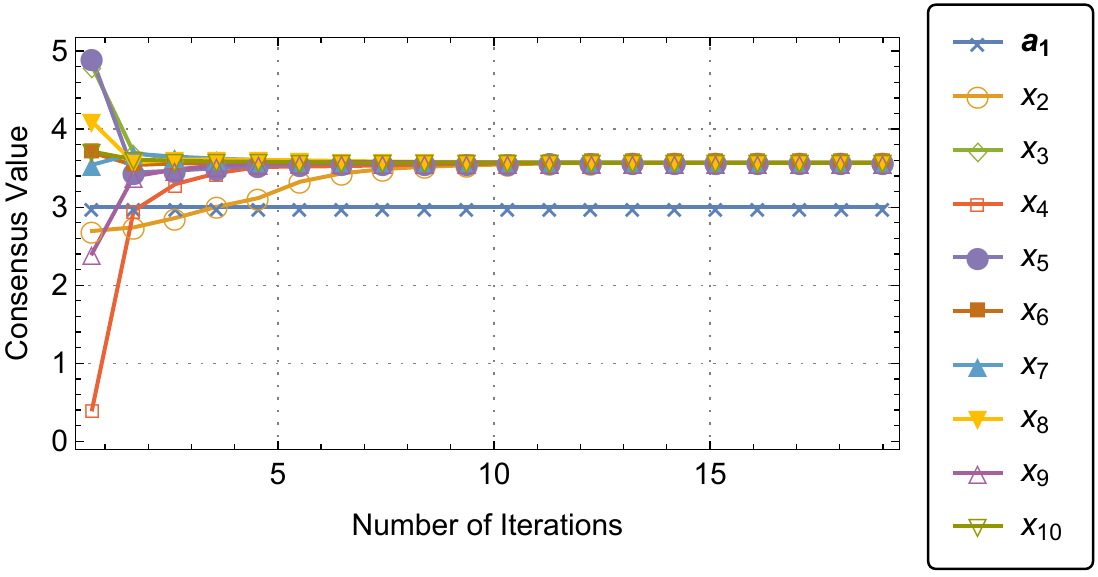}
\caption{Consensus of dynamic network $\mathcal G_1^{(k)}$ with agents $\mathcal V_1$, and set of attacked agents $\mathcal A=\{1\}$.}
\label{fig:C4} 
\end{figure}

%In the second example, we consider that the dynamic network of agents for time instance $k>0$ is given by 
%$\mathcal G_2^{(k)}=\begin{cases}
%	\mathcal G_C &\text{if } k\text{ is odd}\\
%	\mathcal G_D &\text{otherwise}
%\end{cases}.$ 
%The consensus value of each agent, utilizing the iterative scheme~\eqref{eq:id}, is depicted in Figure~\ref{fig:C5}.
%\begin{figure}[h!] 
%\centering      
%\includegraphics[width=0.61\textwidth]{exDyn2.pdf}
%\caption{Consensus of dynamic network $\mathcal G_2^{(k)}$ with agents $\mathcal V_1$, and set of attacked agents $\mathcal A=\{1\}$.}
%\label{fig:C5} 
%\end{figure}
%\vspace{-0.4cm}
\subsection{Dynamic network with noisy agents}\label{sub:dyn_net_noisy}

Finally, we illustrate the scenario where not only the network of agents evolves with time, but also the attacked agents share different values, which are uniform random variables with a fixed mean value. These is captured in the example depicted in  Fig.~\ref{fig:var1}. % and Fig.~\ref{fig:var2}.
\begin{figure}[h!] 
\centering      
\includegraphics[width=0.61\textwidth]{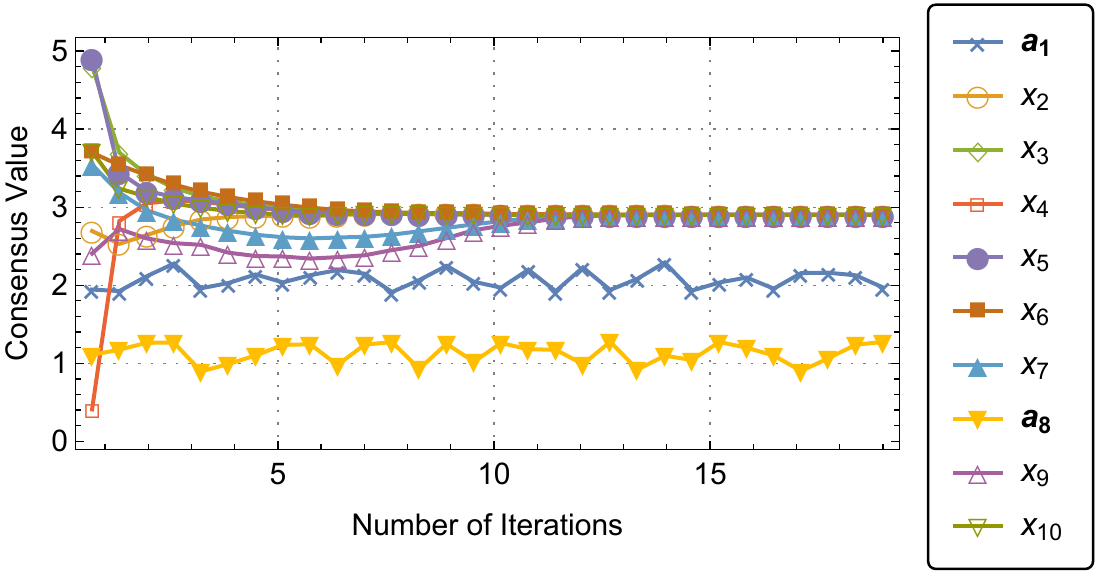}
\caption{Consensus of dynamic network $\mathcal G_1^{(k)}$ with agents $\mathcal V_1$, and set of attacked (noisy) agents $\mathcal A=\{1,8\}$.}
\label{fig:var1} 
\end{figure}

%\begin{figure}[h!] 
%\centering      
%\includegraphics[width=0.61\textwidth]{ExDynRand2.pdf}
%\caption{Consensus of dynamic network $\mathcal G_2^{(k)}$ with agents $\mathcal V_1$, and set of attacked (noisy) agents $\mathcal A=\{1,8\}$.}
%\label{fig:var2} 
%\end{figure}

\begin{figure}[h!] 
\centering      
\includegraphics[width=0.61\textwidth]{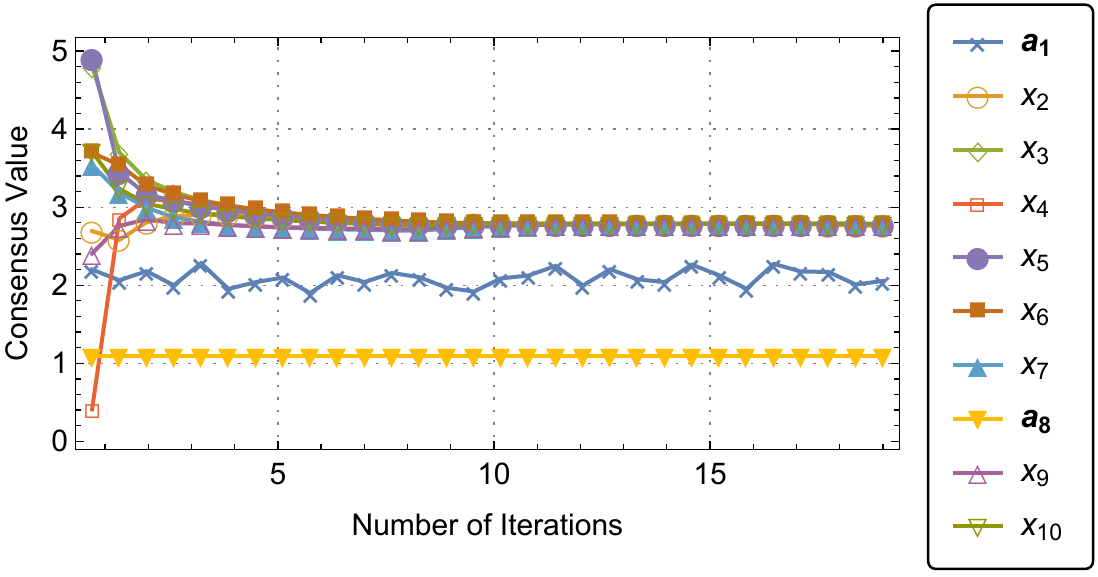}
\caption{Consensus of dynamic network $\mathcal G_1^{(k)}$ with agents $\mathcal V_1$, and set of attacked agents $\mathcal A=\{1,8\}$, where agent $1$ behaves as a noisy node and agent $8$ as an attacked node.}
\label{fig:var} 
\end{figure}
%{\color{red}
%\vspace{-0.3cm}
\subsection{Stochastic communication}
When the communication between agents has a stochastic nature, we may still successfully apply \textsf{RepC}. 
This is illustrated in the next example. We consider the  network $\mathcal G_E$ in Fig.~\ref{fig:gStoc}, with $\mathcal V_1$, and the set of attacked agents $\mathcal A=\{1\}$. Further, at each time step, only a random subset of agents communicate between them. The described situation is depicted in Fig.~\ref{fig:stoc}, where the regular agents could effectively detect the attacked node and achieve the true consensus of the network. 
\begin{figure}[h!] 
\centering      
\includegraphics[width=0.46\textwidth]{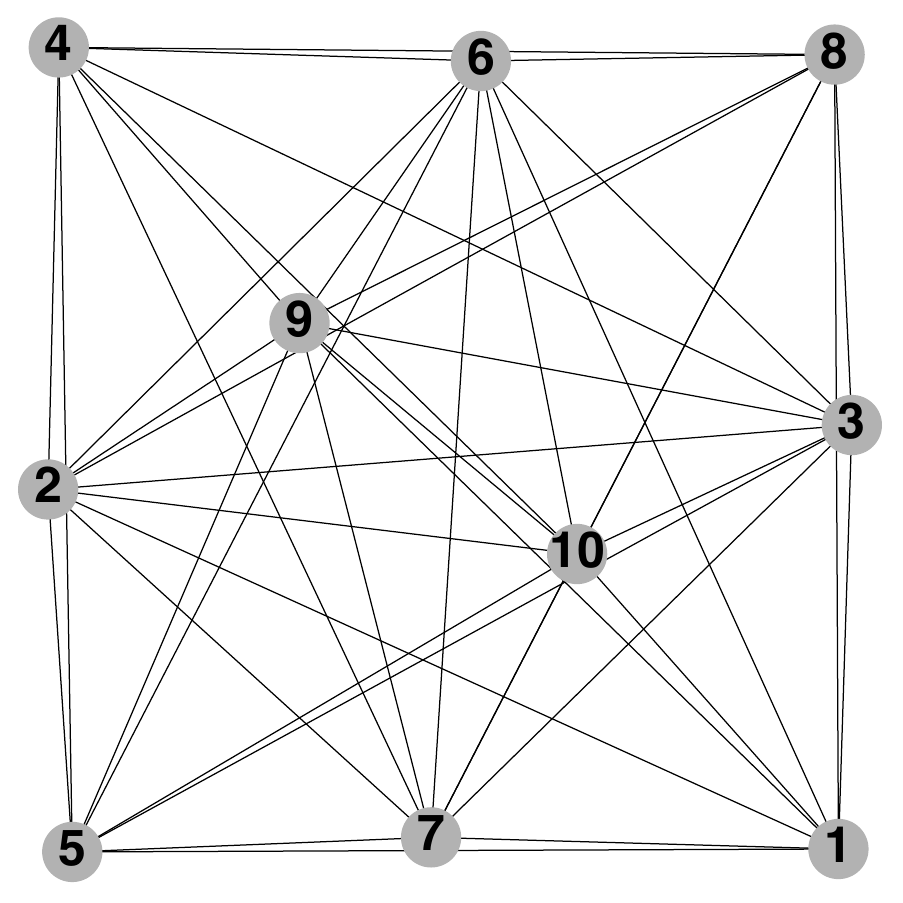}
\caption{Network of agents $\mathcal G_E$.}
\label{fig:gStoc} 
\end{figure}
\begin{figure}[h!] 
\centering      
\includegraphics[width=0.61\textwidth]{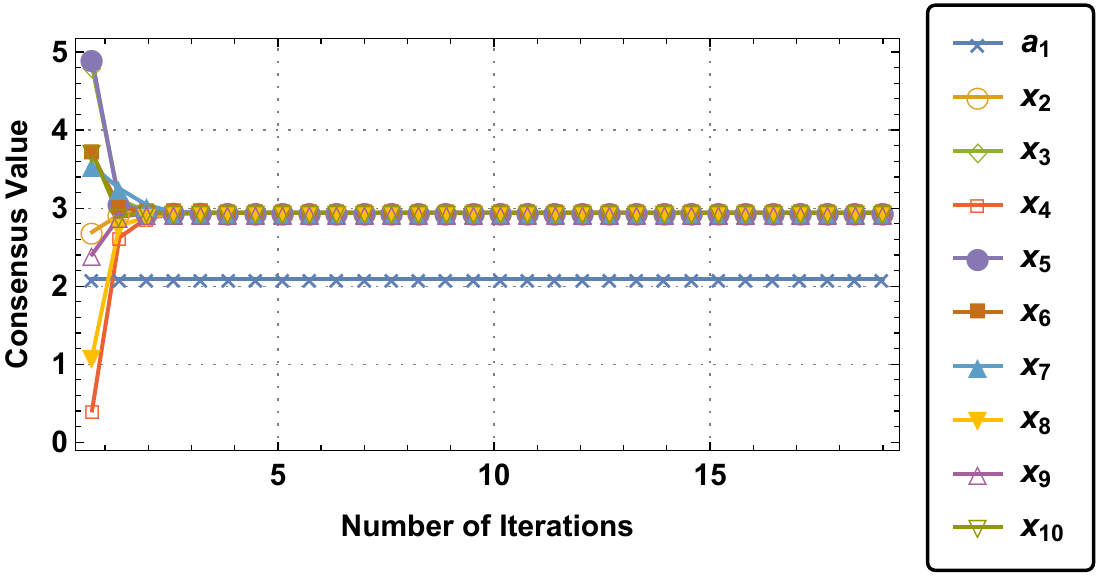}
\caption{Consensus using a \textsf{RepC} method, with network $\mathcal G_E$, set of agents $\mathcal V_1$, set of attacked agents $\mathcal A=\{1\}$ and stochastic communication.}
\label{fig:stoc} 
\end{figure}

\subsection{\textsf{RepC} vs. state-of-the-art}
Here, we illustrate how the proposed algorithm  competes with the state-of-the-art approaches, based on the idea that each agent discards a set of maximum and minimum neighbor values. 

In the next examples, we use the two networks depicted in Fig.~\ref{fig:g1StateArt1}.

%\vspace{-0.25cm}
\begin{figure}[h!] 
\centering      
\subfigure[Network of agents $\mathcal G_F$.]{\label{fig:a}\includegraphics[width=0.25\textwidth]{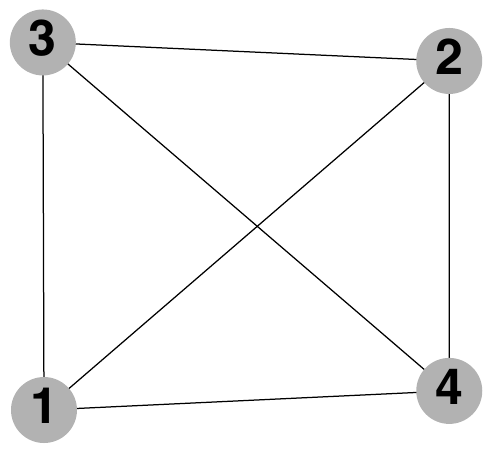}}
%\hfill
\subfigure[Network of agents $\mathcal G_G$.]{\label{fig:b}\includegraphics[width=0.48\textwidth]{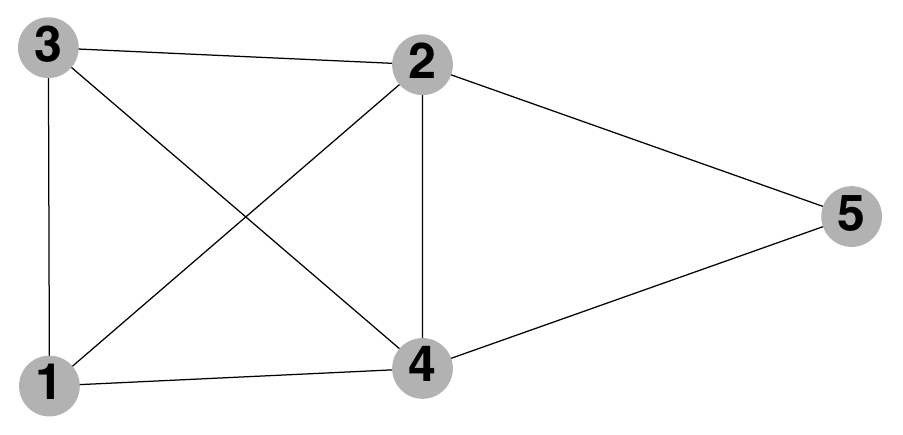}}
%\includegraphics[width=0.25\textwidth]{exDynG2.pdf}
%\caption{Network of agents $\mathcal G_D$.}
\caption{}
\label{fig:g1StateArt1} 
\end{figure}

%\begin{figure}[h!] 
%\centering      
%\includegraphics[width=0.14\textwidth]{gStateArt1.pdf}
%\caption{Network of agents $\mathcal G_F$.}
%\label{fig:g1StateArt1} 
%\end{figure}
%
%
%\begin{figure}[h!] 
%\centering      
%\includegraphics[width=0.26\textwidth]{gStateArt2.pdf}
%\caption{Network of agents $\mathcal G_G$.}
%\label{fig:g1StateArt2} 
%\end{figure}

In the first example, consider the set of agents $\mathcal V_2=\{1,2,3,4\}$, with the complete network (Fig.~\ref{fig:g1StateArt1}~(a)) and attacked agents $\mathcal A=\{1\}$. 

Using the state-of-the-art, i.e., when each agent discards the maximum and minimum neighbors' values, we obtain the result depicted in Fig.~\ref{fig:CStateArt1}. The method is not able to deter the attack and the regular agents converge to the attacker value. 

\begin{figure}[h!] 
\centering      
\includegraphics[width=0.61\textwidth]{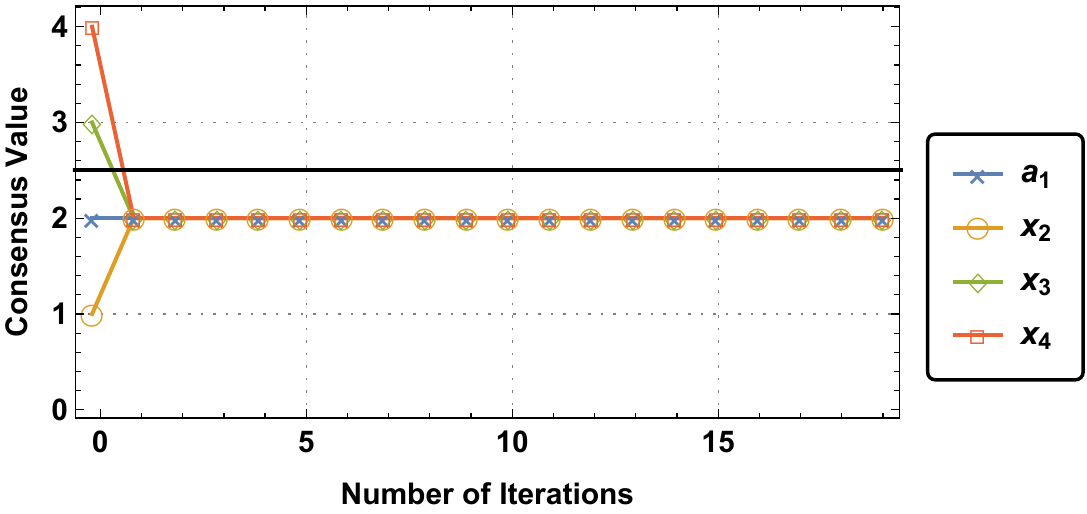}
\caption{Consensus using a \textbf{state-of-the-art} method, with network $\mathcal G_F$, set of agents $\mathcal V_2$, and set of attacked agents $\mathcal A=\{1\}$. The black line is the true consensus value.}
\label{fig:CStateArt1} 
\end{figure}

Using \textsf{RepC}, as illustrated in Fig.~\ref{fig:COur1}, the regular agents converge to a value close to the true value, with a small deviation caused by the  influence of the $\varepsilon$ parameter.  

\begin{figure}[h!] 
\centering      
\includegraphics[width=0.61\textwidth]{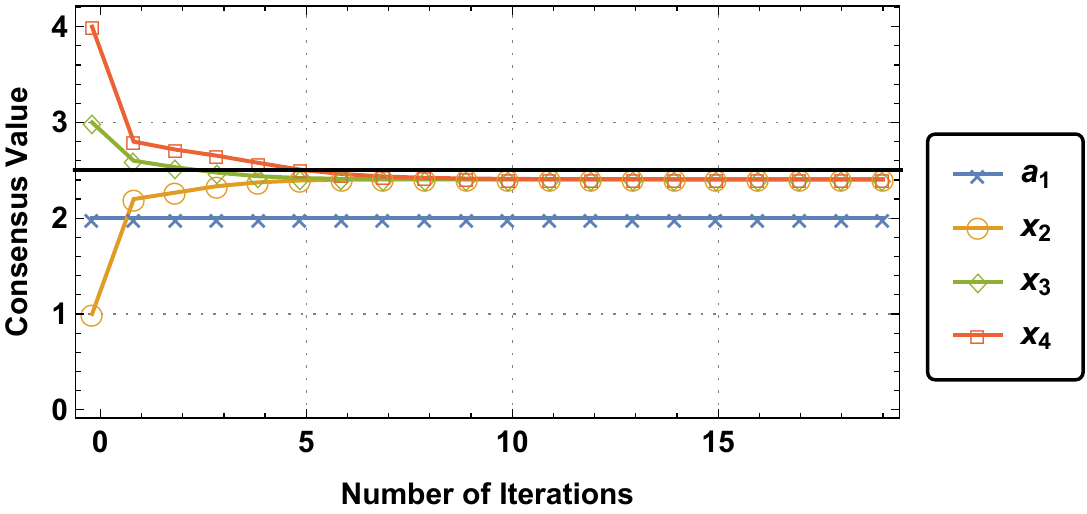}
\caption{Consensus using \textsf{RepC}, with network $\mathcal G_F$, set of agents $\mathcal V_2$, and set of attacked agents $\mathcal A=\{1\}$. The black line is the true consensus.}
\label{fig:COur1} 
\end{figure}

In the second example, we consider the network of agents depicted in Fig.~\ref{fig:g1StateArt1}~(b), the set of agents $\mathcal V_3=\{1,2,3,4,5\}$ and attacked agents set $\mathcal A=\{1\}$. 
The example portraits the scenario where an attacker stubbornly sends to the neighbors the true consensus value.% of the network. 

In Fig.~\ref{fig:CStateArt2}, we present the consensus states of the agents when using the state-of-the-art approach. We can see that the agents are not able to converge to the true consensus value. 

\begin{figure}[h!] 
\centering      
\includegraphics[width=0.61\textwidth]{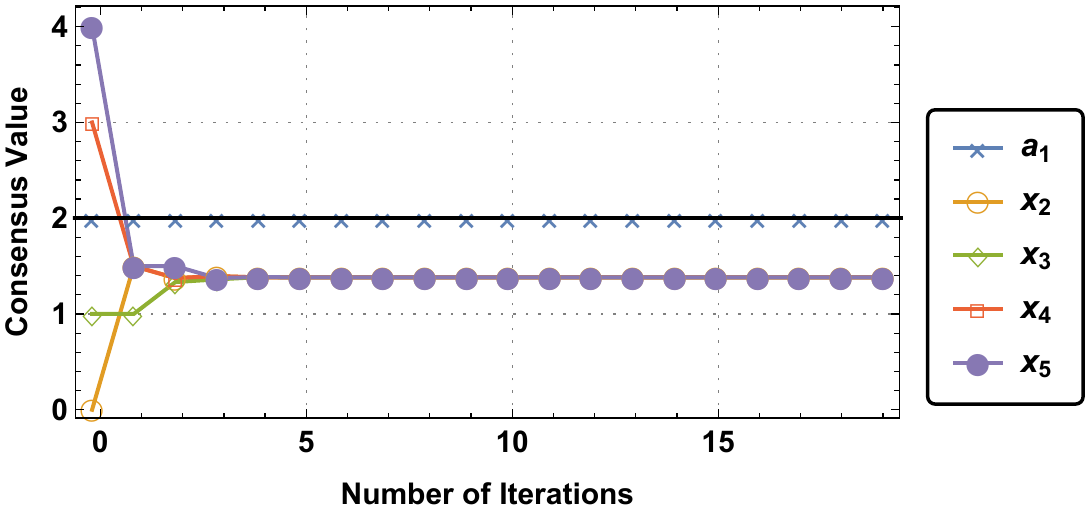}
\caption{Consensus using a \textbf{state-of-the-art} method, with network $\mathcal G_G$, set of agents $\mathcal V_3$, and set of attacked agents $\mathcal A=\{1\}$. The black line is the true consensus value.}
\label{fig:CStateArt2} 
\end{figure}

Subsequently, we present the consensus state of the agents when using \textsf{RepC}. In this case, the agents are able to converge to the true consensus of the network.
 
\begin{figure}[h!] 
\centering      
\includegraphics[width=0.61\textwidth]{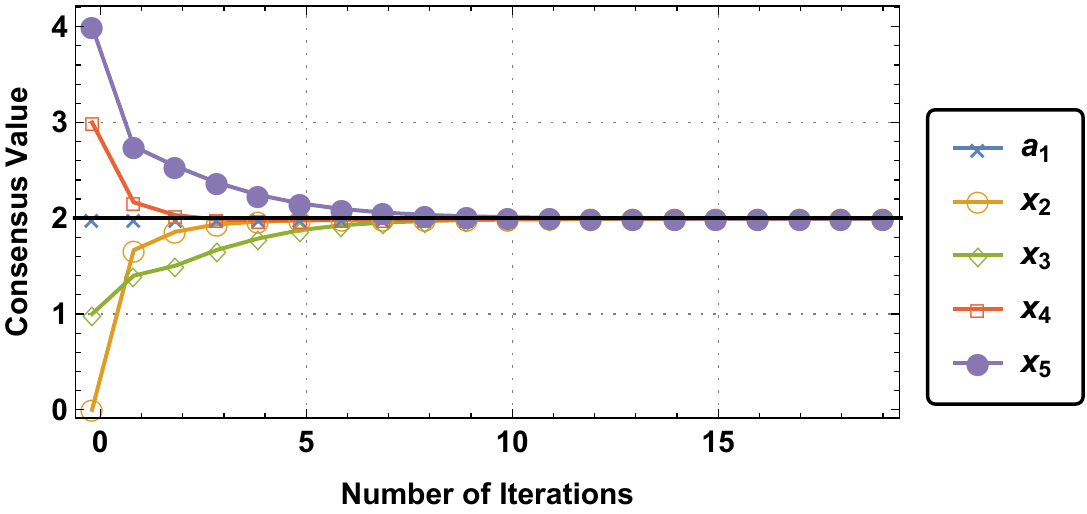}
\caption{Consensus using \textsf{RepC}, with network $\mathcal G_F$, set of  agents $\mathcal V_3$, and set of attacked agents $\mathcal A=\{1\}$. The black line is the true consensus.}
\label{fig:COur2} 
%\vspace{-0.3cm}
\end{figure}

%}

%\subsection{Effect of parameter $\varepsilon$}
%
%Algorithm \textsf{RepC} makes use of a parameter $\varepsilon>0$. The objective of this parameter is to allow setting the algorithm robustness $f$ to a value greater than the number of attacked agents, while enabling each agent to identify the attacked agents among its neighbors, without false positives. 
%Intuitively, the value of $\varepsilon$ should be close to zero, otherwise the final consensus values of the normal agents may deviate from the desired value. 
%To study this effect, we consider a simple attacked agent that always shares the same value to its neighbors.  
%This effect is exactly what is captured in Figure~\ref{fig:epsilon}. When $0<\varepsilon<0.3$, the consensus state of each normal agent, after 20 iterations of algorithm \textsf{RepC}, is very close to the desired consensus value. As $\varepsilon$ grows, the normal agents start to deviate the consensus value from the desired one, and their consensus state changes towards the attacked agent value.
%
%

%{\color{blue}
\subsection{Consensus final error} 

To explore how different is the final consensus value produced by \textsf{RepC} and the consensus value without attacked nodes, we use the complete network of $5$ agents depicted in Fig.~\ref{fig:GA}~(a), with agents' initial states  
$x^{(0)}=\left[\begin{smallmatrix} 1 & 0 & 3 & 1.2 & 2.5 \end{smallmatrix}\right]^\intercal$, where agent $1$ is under attack and shares values from a Gaussian noise with mean $\mu$ and standard deviation $\sigma$. 
The consensus value, without attacked nodes, is $1.489$. 
We compute the absolute difference between the consensus value found with \textsf{RepC} in the non-attacked case and the consensus value obtained with \textsf{RepC} when the attacker follows the mentioned strategy. 
Moreover, we ranged $\mu$ from $0$ to $1$ in steps of $0.005$ and ranged $\sigma$ from $0.1$ to $1$ in steps of $0.005$, repeating each attacking scenario $20$ times to compute the absolute average error. 
The results of the experiment are depicted in Fig.~\ref{fig:error}. 
\begin{figure}[H] 
\centering      
\includegraphics[width=0.46\textwidth]{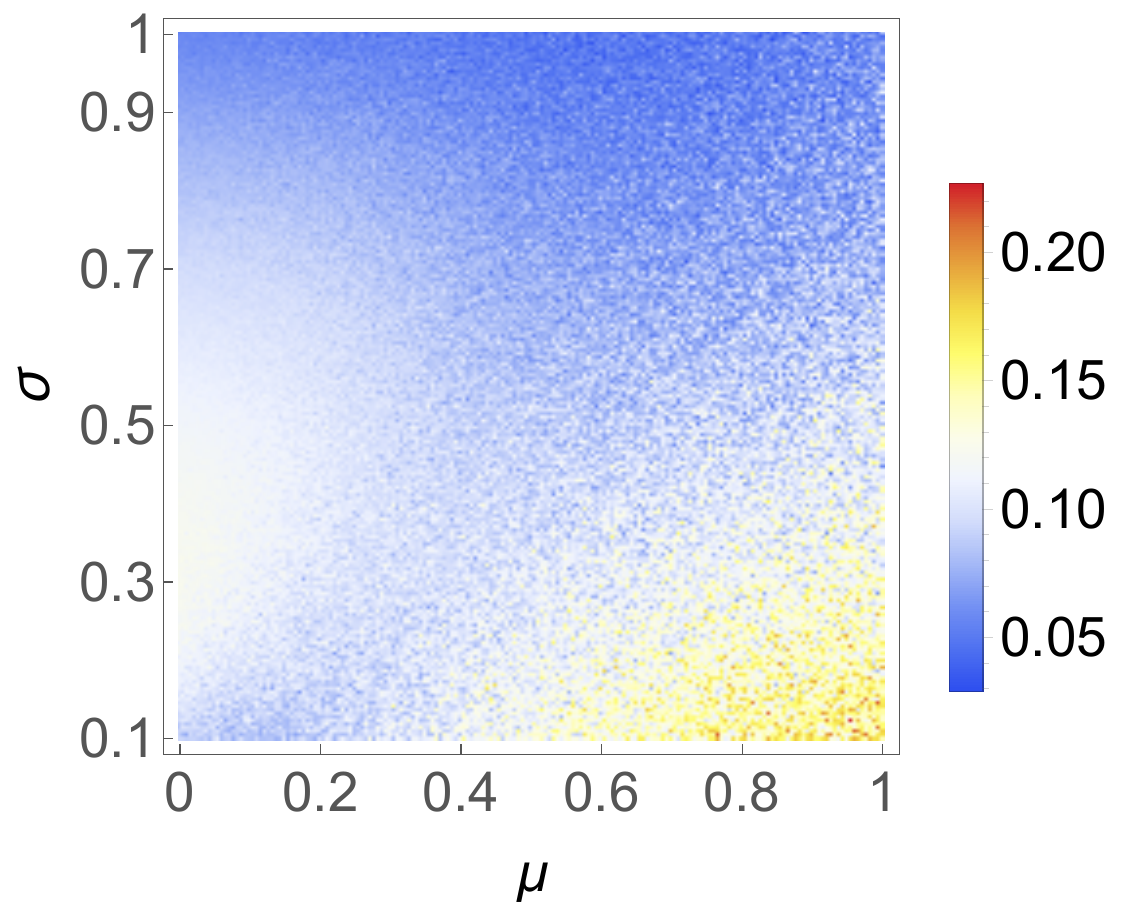}
\caption{Absolute difference between the consensus resulting from Algorithm \textsf{RepC} when node one $1$ is under attack to share the a Gaussian noise with mean $\mu$ and standard deviation $\sigma$.}
\label{fig:error}  
\end{figure} 

We can see from Fig.~\ref{fig:error} that, in average, we obtain a small final consensus error. When $\mu$ is close to $1$ and $\sigma$ is close to $0.1$, the attacked node state value is close to what would be when in the non-attacked scenario ($x_1^{(0)}=1$) and it takes more time to be classified as an attacker by its neighbors, yielding a slightly larger final consensus error. 

%}

%\vspace{-1cm}
\section{Conclusions}\label{sec:conc}
In this work, we presented a reputation-based consensus algorithm (\textsf{RepC}) for discrete-time synchronous and asynchronous communications in a possibly dynamic networks of agents.  
By assigning a reputation value to each neighbor, an agent may discard information from neighbors presenting an abnormal behavior. 

Algorithm \textsf{RepC} converges with exponential rate and it has polynomial time complexity. More specifically, for a network of agents, if we run $i\in\mathbb N$ iterations of \textsf{RepC}, we incur in $\mathcal O(l^2i)$ time complexity, where $l$ is the greatest number of neighbors a node has in that network. 
For attacks with certain properties, we proved that the algorithm does not produce false positives. For other types of attacks, we illustrate the behavior of the proposed algorithm, which also worked as envisaged. 

Future work directions include extending the algorithm for continuous-time  consensus and to introduce the reputation idea for other types of consensus algorithms. 
Furthermore, an important additional theoretical property to prove (even if only for some sorts of attacks) is whether or not \textsf{RepC} may cause an agent to wrongly classify a neighbor as attacked, i.e., if there are false negatives.

{\small
\bibliographystyle{alpha}
\bibliography{acc2016}
}

\end{document}